\documentclass[sigconf]{acmart}
\usepackage{url}
\usepackage{comment}
\usepackage{graphicx}
\usepackage{subfig}
\usepackage{color}
\usepackage{hyperref}
\usepackage{url}

\usepackage{setspace}
\usepackage{lipsum}
\usepackage{enumitem}

\usepackage{caption}
\usepackage[all=normal,paragraphs=tight]{savetrees}

\usepackage{algorithm}
\usepackage{algpseudocode}

\usepackage{amsmath}

\DeclareMathOperator*{\argmin}{arg\ min}
\usepackage{multirow}
\usepackage[mathscr]{euscript}
\usepackage{mathtools}

\DeclarePairedDelimiter\abs{\lvert}{\rvert}%
\DeclarePairedDelimiter\norm{\lVert}{\rVert}%

\theoremstyle{definition}
\newtheorem{definition}{Definition}[section]
\newtheorem{theorem}{Theorem}
\newtheorem{lemma}[theorem]{Lemma}

\usepackage{pifont}
\newcommand{\cmark}{\text{\ding{52}}}
\newcommand{\xmark}{\text{\ding{56}}}

\newcommand{\socialtprdecrease}{$5.4\%$}
\newcommand{\socialfprincrease}{$0.1\%$}

\newcommand*{\red}[1]{\textcolor{red}{#1}}



\copyrightyear{2021}
\acmYear{2021}
\setcopyright{acmlicensed}
\acmConference[CCS '21]{Proceedings of the 2021
ACM SIGSAC Conference on Computer and Communications
Security}{November 15--19, 2021}{Virtual Event, Republic of Korea}
\acmBooktitle{Proceedings of the 2021 ACM SIGSAC Conference on Computer
and Communications Security (CCS '21), November 15--19, 2021, Virtual Event,
Republic of Korea}
\acmPrice{15.00}
\acmDOI{10.1145/3460120.3484776}
\acmISBN{978-1-4503-8454-4/21/11}

\begin{document}
\fancyhead{}

\title{Learning Security Classifiers with \\ Verified Global Robustness Properties}

\author{Yizheng Chen}
\affiliation{%
  \institution{UC Berkeley}
  \country{}
}
\email{1z@berkeley.edu}

\author{Shiqi Wang}
\affiliation{%
  \institution{Columbia University}
  \country{}
}
\email{tcwangshiqi@cs.columbia.edu}

\author{Yue Qin}
\affiliation{%
  \institution{Indiana University Bloomington}
  \country{}
}
\email{qinyue@iu.edu}

\author{Xiaojing Liao}
\affiliation{%
  \institution{Indiana University Bloomington}
  \country{}
}
\email{xliao@iu.edu}

\author{Suman Jana}
\affiliation{%
  \institution{Columbia University}
  \country{}
}
\email{suman@cs.columbia.edu}

\author{David Wagner}
\affiliation{%
  \institution{UC Berkeley}
  \country{}
}
\email{daw@cs.berkeley.edu}

\begin{abstract}

Many recent works have proposed methods to train classifiers with local robustness properties,
which can provably eliminate classes of evasion attacks for most inputs, but not all inputs.
Since data distribution shift is very common in security applications, e.g., often observed
for malware detection, local robustness cannot guarantee that the property holds for
unseen inputs at the time of deploying the classifier. Therefore, 
it is more desirable to enforce global robustness properties
that hold for all inputs, which is strictly stronger
than local robustness.

In this paper, we present a framework and tools for training classifiers that satisfy
global robustness properties. We define new notions of global robustness
that are more suitable for security classifiers.
We design a novel booster-fixer training framework to enforce
global robustness properties. We structure our classifier as an ensemble of logic
rules and design a new verifier to verify the properties. In our training algorithm,
the booster increases the classifier's capacity,
and the fixer enforces verified global robustness properties following
counterexample guided inductive synthesis.

We show that we can train classifiers to satisfy different global robustness properties for
three security datasets, and even multiple properties at the same time, with modest impact
on the classifier's performance. For example, we train a Twitter spam account classifier
to satisfy five global robustness properties, with \socialtprdecrease{} decrease in true positive rate,
and \socialfprincrease{} increase in false positive rate, compared to a baseline XGBoost model
that doesn't satisfy any property.

\end{abstract}

\maketitle

\section{Introduction}

Machine learning classifiers can achieve high accuracy to detect malware,
spam, phishing, online fraud, etc., but they are brittle against evasion attacks.
For example, to detect whether a Twitter account is
spamming malicious URLs, many research works have proposed to use
content-based features, such as the number of tweets containing URLs from that account~\cite{lee2010uncovering,lee2013warningbird,ma2009identifying,
thomas2011design}. These features are useful at achieving high accuracy,
but attackers can easily modify their behavior to evade the classifier.

In this paper, we develop a framework and tools for addressing this problem.
First, the defenders identify a property $\varphi$ that the classifier should satisfy;
typically, $\varphi$ identifies a class of evasion strategies that might be available to an adversary, and specifies a requirement on their effect on the classifier (e.g., that they won't change the classifier's output too much).
Next, the defenders train a classifier $\mathcal{F}$ that satisfies $\varphi$.
We identify several properties $\varphi$ that capture different notions of classifier robustness.
Then, we design an algorithm for the classifier design problem:
\begin{quote}
    \emph{Given a property $\varphi$ and a training set $\mathcal{D}$, train a classifier $\mathcal{F}$ that satisfies $\varphi$.}
\end{quote}
Our algorithm trains a verifiably robust classifier: we can formally verify that $\mathcal{F}$ satisfies $\varphi$.




Existing works focus on training and verifying local robustness properties of classifiers.
Typically, they verify the following local robustness property:
define $\varphi(x)$ to be the assertion that for all $x'$, if $\norm{x'-x}_p \leq \epsilon$, then $\mathcal{F}(x')$ is classified the same as $\mathcal{F}(x)$.
The past works provide a way to verify whether $\varphi(x)$ holds, for a fixed $x$, and then devise ways to train a classifier $\mathcal{F}$ so that $\varphi(x)$ holds for most $x$, but not all $x$ (local robustness).
So far, the most promising defenses against adversarial examples
all take this form, including adversarial training~\cite{madry2017towards},
certifiable training~\cite{dvijotham2018training,mirman2018differentiable,wong2018provable,wang2018mixtrain}, and randomized smoothing techniques~\cite{li2018second,lecuyer2019certified,cohen2019certified}.
To evaluate the local robustness of a trained model, we can
measure the percentage of data points $x$ in the test set that satisfy $\varphi(x)$~\cite{lomuscio2017approach,katz2017reluplex,huang2017safety,ehlers2017formal,fischetti2017deep,dutta2018output,raghunathan2018certified,weng2018towards,gehrai,reluval2018,singh2019boosting,singh2019abstract, shiqi2018efficient,wong2018provable,raghunathan2018semidefinite,tjeng2017evaluating,dvijotham2018dual}.
This measurement is useful if at the time of deploying the classifier,
the real-world data follow the same distribution as the measured test set.
Unfortunately, data distribution shift is very common in security applications (e.g., malware detection~\cite{pendlebury2019tesseract,allix2015your,miller2016reviewer}), so local robustness cannot guarantee
that the robustness property holds for most inputs at deployment time.
In fact, in some cases, the adversary may be able to adapt their behavior by choosing novel samples $x$ that don't follow the test distribution and such that $\varphi(x)$ doesn't hold.

In this paper, we address these challenges by showing how to train classifiers
that satisfy verified global robustness properties.
We define a global robustness property as a universally quantified
statement over one or more inputs to the classifier, and its corresponding outputs: e.g., $\forall x . \varphi(x)$ or $\forall x,x' . \varphi(x,x')$.
Since global robustness holds for all inputs,
it is strictly stronger than local robustness, and it ensures robustness even under distribution shift.
Then, we identify a number of robustness properties that may be useful in security applications, and we develop a general technique to achieve the properties. Our technique can handle
a large class of properties, formally defined in Section~\ref{sec:Supported Properties}.
The vast majority of past work has focused on $\ell_p$ robustness, perhaps motivated by computer vision; however, in security settings, such as detecting malware, online fraud, or other attacks, other notions of robustness may be more appropriate.

There are many challenges in training classifiers with global robustness properties.
First, it is hard to maintain good test accuracy since the definition of
global robustness is much stronger than local robustness.
To the best of our knowledge, among global robustness properties,
only two properties have been previously achieved.
One of them is monotonicity ~\cite{incer2018adversarially,wehenkel2019unconstrained}; and the other is a concurrent work that has proposed $\ell_p$-norm robustness with the option to abstain on non-robust inputs~\cite{leino2021globally}.
For example, researchers have trained a monotonic malware classifier
to defend against evasion attacks that add content to a malware~\cite{incer2018adversarially}.
Monotonicity is useful: it limits the attacker to more expensive evasion operations that may remove malicious functionality from the malware, if they want to evade the classifier.
However, monotonicity is not general enough to capture some types of evasion.
Second, it is challenging to train classifiers with guarantees of global robustness.
Several training techniques sacrifice global robustness
in their algorithms.
For example, DL2~\cite{fischer2018dl2}  proposed several
global robustness properties, but
their training technique only achieves local robustness and cannot learn classifiers with global properties, because they rely on adversarial training.
ART~\cite{lin2020art} presents an abstraction refinement
method to train neural networks with global robustness properties.
In principle ART can guarantee global robustness if the correctness loss
reaches zero, however in their experiments the loss never reached zero.

To overcome these challenges,
we design a novel booster-fixer training framework that enforces global robustness.
Our classifier is structured as an ensemble of logic rules---a new architecture that is
more expressive than trees given the same number of atoms and clauses
(formally defined in Section~\ref{sec:Logic Ensemble Classifier})---and
we show how to verify global robustness properties and then how to train them, for these ensembles.
Intuitively, our algorithm trains a candidate classifier with good accuracy (but not necessarily any robustness),
and then we fix the classifier to satisfy global robustness
by iteratively finding counterexamples and repairing them using the Counterexample Guided Inductive Synthesis (CEGIS) paradigm~\cite{solar2006combinatorial}.
Past works for training monotonic classifiers all use specialized
techniques that do not generalize to other properties~\cite{wehenkel2019unconstrained,incer2018adversarially,archer1993application,daniels2010monotone,kay2000estimating,ben1995monotonicity,duivesteijn2008nearest,feelders2010monotone,gupta2016monotonic}.
In contrast, our technique is fully general and can handle a large class of global robustness properties (formally defined in Section~\ref{sec:Supported Properties}); we even show that we can enforce multiple properties
at the same time (Section~\ref{section:Global Robustness Properties}, Section~\ref{sec:Global Robustness Property Evaluation}).

We evaluate our approach on three security datasets: 
cryptojacking~\cite{kharraz2019outguard}, Twitter spam accounts~\cite{lee2011seven}, and Twitter spam
URLs~\cite{kwon2017domain}. Using security domain knowledge and results from measurement studies,
we specify desirable global robustness properties for each classification task.
We show that we can train all properties individually, and
we can even enforce multiple properties at the same time, with a modest impact on the classifier's performance.
For example, we train a classifier to detect Twitter spam accounts while satisfying five global robustness properties; the true positive rate decreases by \socialtprdecrease{} and the false positive rate increases by \socialfprincrease{},
compared to a baseline XGBoost model that doesn't satisfy any robustness property.


Since no existing work can train classifiers with any global robustness property other than monotonicity,
we compare our approach against two types of baseline models:
1) monotonic classifiers, and 2) models trained with
local versions of our proposed properties.
For the monotonicity property, our results show that our method can achieve comparable or better model performance than prior methods that were specialized for monotonicity.
We also verify that we can enforce each global robustness property we consider, which no prior method achieves for any of the other properties.




Our contributions are summarized as follows.
\begin{itemize}
\item We define new global robustness properties
that are relevant for security applications.

\item We design and implement
a general booster-fixer training procedure to train classifiers with verified global robustness properties.

\item We propose a new type of model, logic ensemble, that is well-suited to booster-fixer training. We show how to verify properties of such a model.


\item We are the first to train multiple global robustness properties. We demonstrate that we
can enforce these properties while maintaining high test accuracy for
detecting cryptojacking websites, Twitter spam accounts, and Twitter spam URLs.
Our code is available at \url{https://github.com/surrealyz/verified-global-properties}.

\end{itemize}

\section{Example}

\begin{figure*}[h!]
    \centering
    \includegraphics[width=\textwidth]{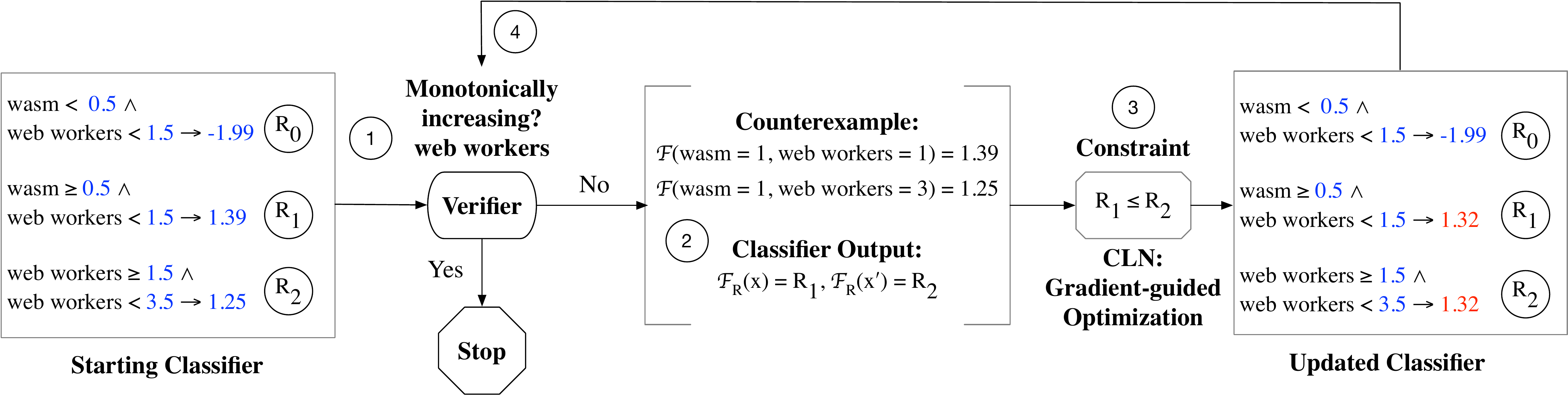}
    \caption{One CEGIS iteration of our training algorithm, illustrated on a simple example. Here we train a classifier to detect cryptojacking, while enforcing a monotonicity property.}
    \label{fig:monotonicity_demo}
\end{figure*}

In this section, we present an illustrative example to show how our training algorithm works.
Within our booster-fixer framework, the fixer follows the
Counterexample Guided Inductive Synthesis (CEGIS) paradigm. 
The key step in each CEGIS iteration is to start from a classifier without the
global robustness property, use a verifier to find counterexamples that violate the property,
and train the classifier for one epoch guided by the counterexample.
This process is repeated until the classifier satisfies the property.
Here, we show how to train one CEGIS iteration for a classifier to detect cryptojacking web pages.
For simplicity, we demonstrate classification using only two features.

\subsection{Monotonicity for Cryptojacking Classifier}
We use two features to detect cryptojacking: whether
the website uses WebAssembly (wasm), and the number of web workers used
by the website~\cite{kharraz2019outguard}. Cryptojacking websites
need the high performance provided by WebAssembly
and often use multiple web worker threads to mine cryptocurrency concurrently.
We enforce a monotonicity property for the web workers feature:
the more web workers a website uses, the more suspicious it should be rated by the classifier, with all else held equal.

Figure~\ref{fig:monotonicity_demo} shows a single CEGIS iteration that starts from
a classifier that violates the property, uses a counterexample to guide the training,
and arrives at an updated classifier that satisfies the property.
The classifier is structured as an ensemble of logic rules.
For example, ``$\texttt{wasm} < 0.5 \land \texttt{web workers} < 1.5 \rightarrow -1.99$"
means that if the website does not use WebAssembly, and has at most one web worker,
the clause adds $R\_0$ to the final prediction value, which is currently $-1.99$.
Otherwise, the clause is inactive and adds nothing to the final prediction value. The colored variables are learnable parameters.
The classifier computes the final score as a sum over all active clauses; if this score is greater than or equal to 0, the webpage is classified as malicious.

Our training procedure executes the following steps:

Step \textcircled{1}, Figure~\ref{fig:monotonicity_demo}: We use formal methods to verify whether the current classifier satisfies the monotonicity property
for the web workers feature. If the property is verified, we have learned a robust classifier
and the iteration stops. Otherwise, the verifier produces a counterexample that violates
the property.

Step \textcircled{2}: We use clause return variables to represent the counterexample.
The counterexample found by the verifier is $x = (\texttt{wasm} = 1, \texttt{web workers} = 1)$,
$x^\prime = (\texttt{wasm} = 1,$ $\texttt{web workers} = 3)$,
such that $x < x^\prime$ and $\mathcal{F}(x) > \mathcal{F}(x^\prime)$.
We compute variables
$\mathcal{F}_R(x) = R_1$, $\mathcal{F}_R(x^\prime) = R_2$
as the classifier output for each input, 
using the sum of return variables from the true clauses.


Step \textcircled{3}: We construct a logical constraint to represent that the counterexample from this pair of samples $(x,x')$ should no longer violate the monotonicity property, i.e., that $\mathcal{F}_R(x) \leq \mathcal{F}_R(x^\prime)$. 
Here, this is equivalent to $R_1 \leq R_2$.
Then, we re-train the classifier subject to the constraint that $R_1 \leq R_2$.
To enforce this constraint, we smooth the discrete classifier using Continuous Logic Networks (CLN)~\cite{ryan2019cln2inv, yao2020learning},
and then use projected gradient descent with the constraint to train the classifier.
Gradient-guided optimization ensures that this counterexample $(x,x')$ will no longer violate the property and tries to achieve the highest accuracy subject to that constraint.
After one epoch of training, the red parameters are changed by gradient descent
in the updated classifier.

Lastly, we discretize the updated classifier and repeat the process again.
In the second iteration, we query the verifier again  (Step \textcircled{4}).
In this example, the updated classifier from the first iteration satisfies the monotonicity property,
and the process stops.

This simplified example illustrates the key ideas
behind our training algorithm.
Appendix~\ref{appendix:Stability for Twitter Account Classifier} shows another example, illustrating that this process is general and can enforce a large class of properties on the classifier.
We define the properties we can support in Section~\ref{sec:Supported Properties}.

\section{Model Synthesis Problem}
\label{section:Model Synthesis Problem}

In this section, we formulate the model synthesis problem mathematically,
and then propose new global robustness properties based on security domain knowledge.

\subsection{Problem Formulation}
\label{section:Problem Formulation}

Our goal is to train a machine learning classifier $\mathcal{F}$
that satisfies a set of global robustness properties.
Without loss of generality, we focus on binary classification
in the problem definition; this can be extended to the multi-class scenario.
The classifier $\mathcal{F}_\theta: \mathbb{R}^n \rightarrow \mathbb{R}$
maps a feature vector $x = [x_1, x_2, ..., x_n]$ with $n$ features to a real number.
Here $\theta$ represent the trainable parameters of the classifier; we omit them from the notation when they are not relevant.
The classifier predicts $\hat{y} = 1$ if $\mathcal{F}(x) \geq 0$,
otherwise $\hat{y} = 0$.
We use $\mathcal{F}(x)$ to represent
the classification score, and
$g(\mathcal{F}(x))$ to denote the normalized
prediction probability for the positive class,
where $g: \mathbb{R} \rightarrow [0, 1]$.
For example, we can use sigmoid
as the normalized prediction function $g$.
We formally define the model synthesis problem here.

\begin{definition}[Model Synthesis Problem]
A model synthesis problem is a tuple $(\Phi, \mathcal{D})$, where
\begin{itemize}
\item $\Phi$ is a set of global robustness properties, $\Phi = \{ \varphi_1, \varphi_2, ..., \varphi_k\}$.
\item $\mathcal{D}$ is the training dataset containing $m$ training samples with their labels
${(x^{(1)}, y^{(1)}), \ldots, (x^{(m)}, y^{(m)})}$.
\end{itemize}
\end{definition}

\begin{definition}[Solution to Model Synthesis Problem]
A solution to the model synthesis problem $(\Phi, \mathcal{D})$
is a classifier $\mathcal{F}_\theta$ with weights $\theta$ that minimizes
a loss function $\mathcal{L}$ over the training set, subject to the requirement
that the classifier satisfies the global robustness properties $\Phi$.
\begin{align}
\begin{split}
\theta = \argmin_\theta \sum_{\mathcal{D}} \mathcal{L} (y, g(\mathcal{F}_\theta(x))) \\
\text{subject to }
\forall \varphi_i \in \Phi, \mathcal{F}_\theta \models \varphi_i
\end{split}
\end{align}

\end{definition}

In Section~\ref{sec:Training Algorithm},
we present a novel training algorithm to solve
the model synthesis problem.


\subsection{Global Robustness Property Definition}
\label{section:Global Robustness Properties}

We are interested in global robustness properties
that are relevant for security classifiers.
Below, we define five general properties that allow us to incorporate domain knowledge
about what is considered to be more suspicious, about what kinds of low-cost evasion strategies the attackers can use without expending too many resources, and about the semantics and dependency among features.


\textbf{Property 1 (Monotonicity):} Given a feature $j$,
\begin{equation}
\forall x, x' \in \mathbb{R}^n . [x_j \leq x'_j \land (\forall i \neq j . x_i = x'_i)] \implies
\mathcal{F}(x) \leq \mathcal{F}(x')
\label{eq:mono_increase}
\end{equation}


This property specifies that the classifier is monotonically increasing along
some feature dimension. It is useful to defend against a class of attacks
that insert benign features into malicious instances (e.g., mimicry attacks~\cite{wagner2002mimicry}, 
PDF content injection attacks~\cite{laskov2014practical},
gradient-guided insertion-only attacks~\cite{grosse2016adversarial},
Android app organ harvesting attacks~\cite{pierazzi2020problemspace}).
If we carefully choose features to be monotonic for a classifier,
injecting content into
a malicious instance can only make it look more malicious to the classifier (not less), i.e., these changes can only increase (not decrease) classification score. Therefore, evading the classifier will require the attacker to adopt more sophisticated
strategies, which may incur a higher cost to the attacker; also, in some settings, these strategies can potentially
disrupt the malicious functionality of the instance, rendering it harmless.

A straightforward variant is to require that the prediction score be monotonically decreasing (instead of increasing) for some features.
For example, we might specify that, all else being equal, the more followers a Twitter account has, the less likely it is to be malicious.
It is cheap for an attacker to obtain a fake account with fewer followers, but expensive to buy a fake account with many followers or to increase the number of followers on an existing account.
Therefore, by specifying that the prediction score should be monotonically decreasing in the number of followers, we force the attacker to spend more money
if they wish to evade the classifier by perturbing this feature.


\textbf{Property 2 (Stability):} Given a feature $j$ and a constant $c$,
\begin{equation}
\forall x, x' \in \mathbb{R}^n . [\forall i \neq j . x_i = x'_i] \implies \abs{\mathcal{F}(x) - \mathcal{F}(x')} \le c
\label{eq:stability}
\end{equation}

The stability property states that for all $x, x'$,
if they only differ in the $j$-th feature, the difference between their prediction scores is bounded by a constant $c$.
The stability constant $c$ is effectively a Lipschitz constant for dimension $j$ (when all other features are held fixed), when $x,x'$ are compared using the $L_0$ distance:
\begin{equation*}
\abs{\mathcal{F}(x) - \mathcal{F}(x')} \le c \norm{x_j - x'_j}_0
\end{equation*}

We can generalize the stability property definition to a subset
of features $J$ that can be arbitrarily perturbed by the attacker.
\begin{equation}
\forall x, x' \in \mathbb{R}^n . [\forall i \notin J . x_i = x'_i]
\implies \abs{\mathcal{F}(x) - \mathcal{F}(x')} \le c \norm{x - x'}_0
\label{eq:gen_stability}
\end{equation}

Researchers have shown that constraining the local Lipschitz constant
to be small when training neural networks can increase the robustness against
adversarial examples~\cite{hein2017formal,cisse2017parseval}.
However, existing training methods rely on regularization techniques
and thus achieve only local robustness; they cannot enforce a global Lipschitz constant.
We are interested in the $\ell_0$ distance,
because some low-cost features can be trivially perturbed
by the attacker to evade security classifiers: the attacker can replace the value of those features with any other desired value.
The stability property captures this by allowing the stable feature to be arbitrarily changed.

\textbf{Low-cost Features.}
Some features can have their values arbitrarily replaced without too much difficulty.
We dub these low-cost features, because it does not cost the attacker much to arbitrarily modify the value of these features.
In particular a low-cost feature is one that is trivial to change, i.e. does not require
nontrivial time, effort, and economic cost to perturb.
All other features are called high-cost.
Section~\ref{sec:Datasets and Property Specifications} gives a concrete analysis of which features are low-cost for three
security datasets.

\textbf{Property 3 (High Confidence):} Given a set of low-cost features $J$,
\begin{equation}
\forall x,x' \in \mathbb{R}^n . [\forall i \notin J . x_i = x'_i] \land g(\mathcal{F}(x)) \ge \delta \implies \mathcal{F}(x') \geq 0
\label{eq:high_confidence}
\end{equation}

The high confidence property states that, for any sample $x$ that is classified as malicious
with high confidence (e.g., $\delta = 0.98$), perturbing any low-cost
feature $j \in J$ does not change the classifier prediction from malicious
to benign. Many low-cost features in security applications are useful
to increase accuracy in the absence of evasion attacks, but they can be easily
changed by the attacker.
For example, to evade cryptojacking detection, an attacker could
use an alias of the hash function name, to evade the hash function feature.
This property allows such features to influence the classification if the sample is near the decision boundary, but for samples classified as malicious with high confidence, modifying just low-cost features should not be enough to evade the classifier.
Thus, samples detected with high confidence by the classifier will be immune to such low-cost evasion attacks.

\textbf{Property 3a (Maximum Score Decrease):}
Given a set of low-cost features $J$,
\begin{equation}
\forall x,x' \in \mathbb{R}^n . [\forall i \notin J . x_i = x'_i] \implies \mathcal{F}(x) - \mathcal{F}(x') \le g^{-1}(\delta)
\label{eq:max_decrease}
\end{equation}

Property 3a is stronger than Property 3.
If the maximum decrease of any classification score is bounded by $g^{-1}(\delta)$, then any high confidence classification
score does not drop below zero.
We provide the proof in Appendix~\ref{appendix:proof}.
In Section~\ref{sec:Robust Training Algorithm},
we design the training constraint for Property 3a
in order to train for Property 3 (Table~\ref{tab:training_constraints}).

\textbf{Property 4 (Redundancy):}
Given $M$ groups of low-cost features $J_1, J_2, \ldots, J_M$
\begin{align}
\begin{split}
    \forall x,x' \in \mathbb{R}^n . [\forall i \notin \bigcup_{m=1}^{M} J_m . x_i = x'_i] \land g(\mathcal{F}(x)) \ge \delta \\
    \land \lnot [\forall m = 1, \ldots, M, \exists j_m \in J_m, x_{j_m} \neq x'_{j_m}] \\
    \implies \mathcal{F}(x') \geq 0
\end{split}
\label{eq:redundancy}
\end{align}


If the attacker perturbs multiple low-cost features,
we would like the high confidence predictions from the classifier
to be robust if different groups of low-cost features are not
perturbed at the same time.
In the redundancy property, we identify $M$ groups of low-cost
features, and require that the attacker has to perturb
at least one feature from each group in order to
evade a high confidence prediction.
In other words, this makes each group of low-cost features
redundant of every other group. If we know all the high-cost features
with any one group of low-cost features, all high confidence
predictions are robust.

\textbf{Property 5 (Small Neighborhood):} Given a constant $c$,
\begin{equation}
    \forall x, x' \in \mathbb{R}^n . d(x,x') \leq \epsilon \implies \abs{\mathcal{F}(x) - \mathcal{F}(x')} \le c \cdot \epsilon
\label{eq:small_neighborhood}
\end{equation}
where $d(x,x') = \max_i \{\abs{x_i-x'_i}/\sigma_i \}$.

The small neighborhood property specifies that for any two data points
within a small neighborhood defined by $d$, we want
the classifier's output to be stable.
We define the neighborhood by a new distance metric $d(x, x')$ that measures
the largest change to any feature value, normalized by the standard deviation of that input feature.
$d(x,x')$ is essentially a $\ell_\infty$ norm, applied to normalized feature values.
We chose not to use the $\ell_\infty$
distance directly because different features for security classifiers often have a different scale.


\section{Property Verification}
In this section, we describe the key ingredients we need to solve the model synthesis problem.
We define a new type of classifier that is well-suited to model synthesis,
and a verification algorithm to verify whether the classifier satisfies the properties.

\subsection{Logic Ensemble Classifier}
\label{sec:Logic Ensemble Classifier}

We propose a new type of classification model, which we call a logic ensemble.
We show how to train logic ensemble classifiers that satisfy global robustness properties.

\begin{definition}[Logic Ensemble Definition]
A logic ensemble classifier $\mathcal{F}$ consists of a set of clauses.
Each clause has the form
\begin{onehalfspace} \begin{center}
$B_1(\alpha_1, \beta_1) \land B_2(\alpha_2, \beta_2) \land \dots
\land B_m(\alpha_m, \beta_m) \rightarrow R$
\end{center} \end{onehalfspace}
\noindent
where $B_1 \ldots B_m$ are atoms and $R$ is the activation value of the clause.
Each atom $B_i$ has the form $\alpha_i x_j < \beta_i$ for some $j$.
Here the $\alpha_i, \beta_i, R$ are trainable parameters for the classifier.
The implication denotes that if the body of the clause holds (all atoms $B_1 \dots B_m$ are true),
then the clause returns an activation value $R$, otherwise it returns $0$.
The classifier's output is computed as $\mathcal{F}_{\alpha,\beta,R}(x) = \sum R_i$, where the sum is over all clauses that are satisfied by $x$. 
\end{definition}



Logic ensembles can be viewed as a generalization of decision trees.
Any decision tree (or ensemble of trees) can be expressed as a logic ensemble, with one clause per leaf in the tree,
but logic ensembles are more expressive (for a fixed number of clauses) because they can also represent other structures of rules.
Researchers have previously shown how to train decision trees with monotonicity properties, so our work can be viewed as an extension of this to a more expressive class of classifiers and a demonstration that this allows enforcing other robustness properties as well.

\subsection{Integer Linear Program Verifier}
\label{sec:Integer Linear Program Verifier}

We present a new verification algorithm that uses integer linear programming
to verify the global robustness properties of logic ensembles,
including trees.
First, we encode the logic ensemble using boolean variables,
adding consistency constraints
among the boolean variables. Then, for each global robustness property, we symbolically represent the input and output of the classifier in terms of these boolean variables, and add extra constraints to assert that the robustness property is violated.
Next, we check feasibility of these constraints, expressing them as a 0/1 integer linear program.
If an ILP solver can find a feasible solution, the classifier
does not satisfy the corresponding global robustness property, and the solver will give us a counterexample.
On the other hand, if the integer linear program is infeasible,
the classifier satisfies the global robustness property.

We describe our algorithm in more detail below. We use the binary variables in the 0/1 integer linear program
to represent an arbitrary input $x$:

\emph{Atom ($p$):} We use $p_i$ variables to encode the truth value
of atoms. Each atom is transformed into the same form of predicate $x_j < \eta_i$.
Therefore, each predicate variable $p_i$ is associated with
a feature dimension $j$ and the inequality threshold $\eta_i$.

\emph{Clause Status ($l$):} We use $l_k$ variables to encode
the truth value of clauses. When $l_k = 1$, all atoms in the $k$-th clause
are true, and the clause adds $R_k$ activation value for the classifier output. 

\emph{Auxiliary Variables ($a$):} We use $a_{i1}$ and $a_{i2}$ variables to encode the
neighborhood range for the small neighborhood property,
defined in Equation~\eqref{eq:small_neighborhood}.
For each predicate $x_j < \eta_i$, we create $a_{i_1}$
variable for $x_j < \eta_i- \sigma_j * \epsilon$,
and $a_{i_2}$ variable for $x_j < \eta_i + \sigma_j * \epsilon$.
If $x_j$ is within $[\eta_i- \sigma_j * \epsilon, \eta_i + \sigma_j * \epsilon]$,
we must have $a_{i_2} - a_{i_1} = 1$.

\emph{Double Variables:} All the aforementioned variables
are doubled as $p', l', a'$ to represent the perturbed input $x'$ bounded by
the robustness property definition. The classifier's output for arbitrary $x, x'$
are:
\begin{center}
$\mathcal{F}(x) = \sum\limits_k{l_k R_k}, \mathcal{F}(x') = \sum\limits_k{l'_k R_k}$
\end{center}

\begin{figure*}[h!]
    \centering
    \includegraphics[width=\textwidth]{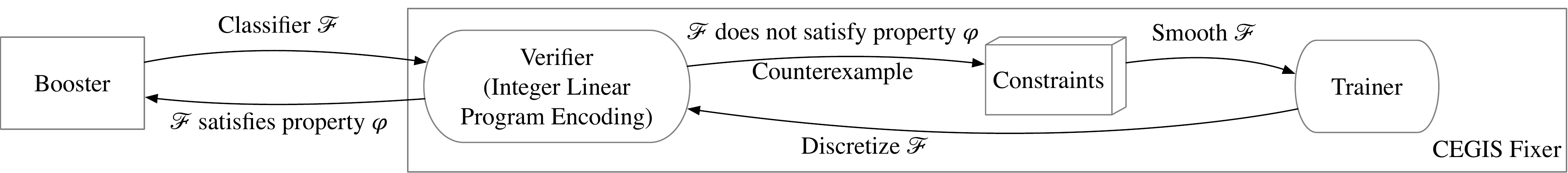}
    \caption{Booster-fixer training framework.}
    \label{fig:framework}
\end{figure*}

Then, we create the following linear constraints to ensure dependency between
variables of the classifier.

\emph{Integer Constraints:} We merge predicates for integer features.
For example, if $x_5$ is an integer feature, we use the same binary variable
to represent atoms $x_5 < 0.2$ and $x_5 < 0.3$.

\emph{Predicate Consistency Constraints:}
Predicate variables for the same feature dimension are sorted and constrained accordingly.
For any $p_i, p_t$ belonging to the same feature dimension with $\eta_i < \eta_t$,
$x_j < \eta_t$ must be true if $x_j < \eta_i$ is true. Thus, we have $p_i \leq p_t$. 

\emph{Redundant Predicate Constraints:} We set redundant variables to be always 0.
For example, $x_j < 0$ is always false for a nonnegative feature.

Lastly, we use standard boolean encoding for \textbf{Property Violation Constraints} to verify a given property,
as shown in Table~\ref{tab:constraints}. Each property has a pair of input
and output constraints. In addition, we encode input consistency constraints
for a given property. 

\emph{Input Consistency Constraints:} For any $x_i$ and $x'_i$,
if they are defined to be the same by the property, we set the related predicate
variables $p$ and $p'$ to have the same value.

\emph{Monotonicity:}
For two arbitrary inputs $x$ and $x'$, if $x < x'$, then there must be at least one more predicate true for $x'$.
The output constraint for (1) denotes violation to monotonically non-decreasing output,
and (2) denotes violation to monotonically non-increasing output.

\emph{Stability:} The input constraint says $x$ and $x'$ are different, and the output constraint
says the difference between $\mathcal{F}(x)$ and $\mathcal{F}(x')$ are larger than the stable constant $c_\text{stability}$.

\emph{High Confidence:} The input constraint says $x$ is classified as malicious with at least $\delta$ confidence.
The output constraints says $x'$ is classified as benign. 

\emph{Redundancy:} The input and output constraints are the same as high confidence property.
However, we encode predicate consistency constraints differently. We set variable equality constraints
such that $x_i = x'_i$ for $i$ outside the $M$ low-cost feature groups.
We encode the disjunction of the conditions that
only features from the same group are changed.

\emph{Small Neighborhood:}
For input constraint, for each $j$, we first encode the conjunction that $x_j$ and $x'_j$
are both within a small neighborhood interval $a_{i_2} - a_{i_1} = 1$.
Then, we encode the disjunction that $x_j$ and $x'_j$ can be only within one
of such intervals surrounding the predicates. The output constraint says the difference
between the outputs are larger than the allowed range.

\begin{table}[t!]
  \centering
  \begin{tabular}{l | c}
    \hline
    \textbf{Property} & \textbf{Property Violation Constraints} \\
    \hline \hline
    \multirow{6}{*}{Monotonicity} & \multirow{3}{*}{
    \begin{tabular}{@{}c@{}}
        (1) In: $\sum\limits_{x}{p_i} \leq \sum\limits_{x'}{p'_i} + 1$, \\
        Out: $\sum\limits_{x}{l_k*R_k} > \sum\limits_{x'}{l'_k*R_k}$
    \end{tabular}
    }   \\
    & \\
    & \\ \cline{2-2}
    & 
    \multirow{3}{*}{
    \begin{tabular}{@{}c@{}}
        (2) In: $\sum\limits_{x}{p_i} \leq \sum\limits_{x'}{p'_i} + 1$, \\
        Out: $\sum\limits_{x}{l_k*R_k} < \sum\limits_{x'}{l'_k*R_k}$
    \end{tabular}
    }   \\
    & \\
    & \\ \hline
    \multirow{3}{*}{Stability} &     \multirow{3}{*}{
    \begin{tabular}{@{}c@{}}
        In: $\abs{\sum\limits_{x}{p_i} - \sum\limits_{x'}{p'_i}} \geq 1$, \\
        Out: $\abs{\sum\limits_{x}{l_k*R_k} - \sum\limits_{x'}{l'_k*R_k}} > c_\text{stability}$
    \end{tabular}
    }   \\
    & \\
    & \\ \hline
    \multirow{3}{*}{\begin{tabular}[c]{@{}c@{}}High \\ Confidence\end{tabular}} & \multirow{3}{*}{
    \begin{tabular}[c]{@{}c@{}}
    In: $\sum\limits_{x}{l_k*R_k} \geq g^{-1}(\delta)$, \\
    Out: $\sum\limits_{x'}{l'_k*R_k} < 0$
    \end{tabular}} \\
    & \\
    & \\ \hline
    \multirow{2}{*}{Redundancy} & \multirow{2}{*}
    {\begin{tabular}[c]{@{}c@{}}
    Same constraints as high confidence.
    \\
    Diff predicate consistency constr.
    \end{tabular}} \\
    & \\ \hline
    \multirow{4}{*}{\begin{tabular}[c]{@{}c@{}}Small \\ Neighborhood\end{tabular}} & In: for each feature $j$, $x_j$ and $x'_j$ are \\
    & in the same interval $[\eta_i - \sigma_j * \epsilon, \eta_i + \sigma_j * \epsilon]$, \\ 
    & Out: $\abs{\sum\limits_{x}{l_k*R_k} - \sum\limits_{x'}{l'_k*R_k}} > \epsilon * c_\text{neighbor} $ \\ \hline
  \end{tabular}
  \caption{Property violation constraints for the verifier.
  }
  \label{tab:constraints}
\end{table}

\section{Training Algorithm}
\label{sec:Training Algorithm}

\subsection{Framework}
\label{sec:Framework}

Figure~\ref{fig:framework} gives an overview of our booster-fixer training framework.
We have two major components, a booster and a fixer, which interact with each other
to train a classifier with high accuracy that satisfies the global robustness properties.

The booster increases the size of the classifier, and improves classification performance.
We run $N$ boosting rounds.
The classifier is a sum of logic ensembles,
$\mathcal{F}(x) = \sum_{b=1}^{N}{f_b(x)}$,
where each $f_b$ is a logic ensemble.
In the $b$-th boosting round, the booster adds $f_b$ to the ensemble,
proposing a candidate classifier $\sum_{i=1}^{b}{f_i(x)}$ (which does not need to satisfy any robustness property); then the fixer fixes property violations for this classifier.
Empirically, more boosting rounds typically lead to better test accuracy after fixing the properties.

The fixer uses counterexample guided inductive synthesis (CEGIS) to fix
the global robustness properties for the current classifier.
We use a verifier and a trainer to iteratively
train the classifier, eliminating counterexamples in each iteration,
until the classifier satisfies the properties.
In each CEGIS iteration, we first use the verifier to find a counterexample
that violates the property. Then, we
use training constraints to eliminate the counterexample.
The training constraints reduce the space of candidate classifiers
and make progress towards satisfying the property.
We accumulate the training constraints over the CEGIS iterations,
so that our classifier is guaranteed to satisfy global robustness properties
when the fixer returns a solution.
After we fix the global robustness properties for the classifier $\sum_{i=1}^{b}{f_i(x)}$, 
we go back to boost the next $b+1$ round, to further improve the test accuracy.
We will discuss the details of our training algorithm next.

\subsection{Robust Training Algorithm}
\label{sec:Robust Training Algorithm}


Algorithm~\ref{alg:training} presents the pseudo-code for our global robustness training algorithm.
As inputs, the algorithm needs specifications of the global robustness properties $\Phi$
(Section~\ref{section:Global Robustness Properties}) and a
training dataset $\mathcal{D}$ to train for both robustness and accuracy.
In addition, we need a booster $\mathcal{B}$ (Section~\ref{sec:Framework}),
a verifier $\mathcal{V}$ (Section~\ref{sec:Integer Linear Program Verifier}),
a trainer $\mathcal{S}$ (described below) and a loss function $\mathcal{L}$
to run the booster-fixer rounds. We can specify the number of boosting rounds $N$.
The algorithm outputs a classifier $\mathcal{F}$ that satisfies all the specified
global robustness properties.

\begin{algorithm}[t!]
\caption{Global Robustness Property Training Algorithm}
\label{alg:training}
\begin{flushleft}
\textbf{Input:} Global robustness properties $\Phi$.\\
Training set $\mathcal{D} = \{(x^{(i)}, y^{(i)})\}$. Number of boosting rounds $N$. \\
\textbf{Input:} Booster $\mathcal{B}$. Verifier $\mathcal{V}$. Trainer $\mathcal{S}$. Loss function $\mathcal{L}$. \\
\textbf{Output:} classifier $\mathcal{F}$ that satisfies all the properties in $\Phi$. \\
\end{flushleft}
\begin{algorithmic}[1]
\State Initialize an empty classifier $\mathcal{F}$.
\State Initialize an empty set of constraints $\mathcal{C}$.
\For{$b=1$ to $N$}
	\State $\mathcal{B}$ adds $f_b$ to $\mathcal{F}$, so that $\mathcal{F}(x) = \sum_{i=1}^{b}{f_i(x)}$.
	\While{$\exists \varphi_i \in \Phi, \mathcal{F} \not\models \varphi_i$}
	\For{each $\varphi_i \in \Phi$}
		\If{$\mathcal{F} \not\models \varphi_i$}
			\State Call $\mathcal{V}(\mathcal{F})$ to get a counterexample $(x, x')$.
			\State Call GenConstraint($x, x'$) to get a constraint $c$.
			\State Add $c$ to $\mathcal{C}$.
		\EndIf
	\EndFor
	\If{the constraints in $C$ are infeasible}
	    \State \Return Failure.
	\EndIf
    \State Update $\theta=(\alpha,\beta,R)$ using $\mathcal{S}(\alpha, \beta, R, \mathcal{D}, \mathcal{C})$.
	\EndWhile
\EndFor \\
\Return $\mathcal{F}$ \\

\Function{GenConstraint}{$x, x'$}:
    \State \Return a constraint on $R$ that implies $\varphi(x,x')$,
    \State when $x,x',\alpha,\beta$ are fixed at their current values. 
\EndFunction \\

\Function{$\mathcal{S}$}{$\alpha, \beta, R, \mathcal{D}, \mathcal{C}$}:
	\State Update $\alpha,\beta,R$ using projected gradient descent:
	\State $\alpha,\beta,R = \argmin_{\alpha,\beta,R} \sum_{\mathcal{D}} \mathcal{L} (y, g(\mathcal{F}_{\alpha,\beta,R}(x)))$
	\State s.t. $R$ satisfies all constraints in $\mathcal{C}$
\EndFunction
\end{algorithmic}
\end{algorithm}

First, our algorithm initializes an empty ensemble classifier $\mathcal{F}$ such that
we can add sub-classifiers into it over the boosting rounds (Line 1).
We also initialize an empty set of constraints $\mathcal{C}$ (Line 2).
Then, we go through $N$ rounds of boosting in the for loop
from Line 3 to Line 18. Within each boosting round $b$, the booster $\mathcal{B}$
adds a tree to the ensemble classifier, such that the current classifier is
$\mathcal{F}(x) = \sum_{i=1}^{b}{f_i(x)}$.
The fixer runs the while loop from Line 5 and Line 17.
As long as the classifier does not satisfy all specified global robustness properties,
we proceed with fixing the properties (Line 5).
For each property, if the model does not satisfy the property,
the verifier $\mathcal{V}$ produces a counterexample $(x, x')$ (Line 8).
Then, we generate a constraint $c$ that can eliminate the counterexample
by calling a procedure GenConstraint$(x, x')$ (Line 9).
We add the constraint to the set $C$.
If the set of constraints are infeasible, the algorithm returns failure.
Otherwise, we use the trainer $\mathcal{S}$
to train the weights $\theta$ using projected gradient descent (Line 16 calls $\mathcal{S}(\alpha, \beta, R, \mathcal{D})$).
We follow the gradient of the loss function w.r.t. the weights $\theta$,
update the weights, and then
we project the weights onto the $\ell_2$ norm ball centered around updated weights, subject to all constraints in $C$,.
Therefore, the weights satisfy all constraints in $C$.

\textbf{Generating Constraint.} The GenConstraint function
generates a constraint according to counterexample $(x, x')$.
We use $\mathcal{F}_R(x)$ to represent the \emph{equivalence class} of $x$: all inputs that are classified
the same as $x$, i.e., their classification score
is a sum of return values for the same set of clauses as $x$.
We can use constraints over
$\mathcal{F}_R(x)$ and $\mathcal{F}_R(x')$ to capture
the change in the classifier's output, to satisfy the global robustness property for all counterexamples in the equivalence class.
Specifically, in Table~\ref{tab:training_constraints}, we list the constraints
for five properties we have proposed. The constraints for monotonicity,
stability, redundancy, and the small neighborhood properties
have the same form as the output requirement specified in
the corresponding property definitions.
For the high confidence property, our training constraint
is to bound the drop of the classification score to be no
more than the $g^{-1}$ of the high confidence threshold $\delta$. This constraint aims to satisfy Property 3a (Equation~\ref{eq:max_decrease}), which then satisfies
Property 3 high confidence (Lemma~\ref{lemma1}).
This constraint eliminates counterexamples
faster than using the constraint $\mathcal{F}_R(x') \geq 0$.

\begin{table}[t!]
  \centering
  \begin{tabular}{l | c}
    \hline
    \textbf{Property} & \textbf{Training Constraints} \\
    \hline \hline
    \multirow{2}{*}{Monotonicity} & (1) $\mathcal{F}_R(x) \leq \mathcal{F}_R(x')$  \\ \cline{2-2}
    & (2) $\mathcal{F}_R(x) \geq \mathcal{F}_R(x')$\\ \hline
    Stability & $\abs{\mathcal{F}_R(x) - \mathcal{F}_R(x')} \leq c_\text{stability}$ \\ \hline
    High Confidence & $\mathcal{F}_R(x) - \mathcal{F}_R(x') < g^{-1}(\delta) $ \\ \hline
    Redundancy & Same as high confidence. \\ \hline
    Small Neighborhood & $\abs{\mathcal{F}_R(x) - \mathcal{F}_R(x')} \leq \epsilon * c_\text{neighbor}$ \\ \hline
  \end{tabular}
  \caption{Constraints used for the training algorithm.
  }
  \label{tab:training_constraints}
\end{table}

\textbf{CLN Trainer.} Within the fixer, we use 
Continuous Logic Networks (CLN)~\cite{ryan2019cln2inv} to
train the classifier to satisfy all constraints in $C$.
If we directly enforce constraints over the weights of the classifier,
the structure and weights will not have good accuracy.
We want to use gradient-guided optimization to preserve accuracy
of the classifier while satisfying the constraints.
Since our discrete ensemble classifier is non-differentiable,
we first use CLN to smooth the logic ensemble.
Following Ryan et al.~\cite{ryan2019cln2inv}, we use a shifted and scaled sigmoid function
to smooth the inequalities, product t-norm to smooth conjunctions.
To train the smoothed model, we use binary cross-entropy loss
as the loss function $\mathcal{L}$ for classification,
and minimize the loss using projected gradient descent
according to the constraints $C$.
After training, we discretize the model back to logic ensemble for prediction,
so we can verify the robustness properties.
Note that although our training constraints $C$ are only related to
the returned activation values of the clauses (Table~\ref{tab:training_constraints}),
the learnable parameters of atoms may change as well due to the projection
(See Appendix~\ref{appendix:Stability for Twitter Account Classifier} for an example).
In some cases, the structure of the atom can change as well.
For example, if an atom $x_0 < 5$ is trained to become $-0.5 * x_0 < 2$,
this changes the inequality of the atom.

\subsubsection{Supported Properties.}
\label{sec:Supported Properties}
Our framework can handle any global robustness property $\varphi$
of the form
$\forall x_1,\dots,x_k . \mu(x_1,\dots,x_k) \implies \nu(\mathcal{F}(x_1),\dots,\mathcal{F}(x_k))$
where the set of values
$\{(y_1,\dots,y_k): \nu(y_1,$
$\dots,y_k)\}$ is a convex
set, as then we can project the classifier weights accordingly
(line 27 to line 29 in Algorithm~\ref{alg:training}).
For example, for the monotonicity property, $k = 2$,
$\mu(x_1, x_2) \coloneqq x_1 \leq x_2$,
and $\nu(\mathcal{F}(x_1),\mathcal{F}(x_2)) \coloneqq \mathcal{F}(x_1) \leq \mathcal{F}(x_2)$.
This class includes but is not limited to all global robustness properties with arbitrary
linear constraints on the outputs of the classifier.


\subsubsection{Algorithm Termination.}
Algorithm~\ref{alg:training} is guaranteed to terminate.
When the algorithm terminates, if it finds a classifier, the classifier is guaranteed to satisfy the properties.
However, there is no guarantee that it will find a classifier (line 14 of Algorithm~\ref{alg:training} returns Failure),
but empirically our algorithm can find an accurate classifier that satisfies all the specified properties,
as shown in the results in Section~\ref{sec:Robust Logic Ensembles}.





\section{Evaluation}


\begin{table*}[!bt]
	\centering
	\small
	\begin{tabular}{c|c|l}
		\hline
		\textbf{Dataset} & \textbf{Property} & \textbf{Specification} \\ \hline\hline
		\multirow{6}{*}{Cryptojacking} 
	    & - & Low-cost features: whether a website uses one of the hash functions on the list. \\ \cline{2-3}
	    & Monotonicity & Increasing: all features \\ \cline{2-3}
		& Stability & All features are stable. Stable constant = $0.1$\\ \cline{2-3}
		& High Confidence & $\delta = 0.98$  \\ \cline{2-3}
		& Small Neighborhood & $\epsilon = 0.2, c = 0.5$ \\ \cline{2-3}
	    & Combined & Monotonicity, stability, high confidence, and small neighborhood \\ \hline
		\multirow{10}{*}{\begin{tabular}[c]{@{}c@{}}Twitter Spam \\ Accounts \end{tabular}} 
		& \multirow{2}{*}{-} & Low-cost features: LenScreenName ($\geq 5$ char), LenProfileDescription, NumTweets, NumDailyTweets, \\ 
		& & TweetLinkRatio, TweetUniqLinkRatio, TweetAtRatio, TweetUniqAtRatio. \\ \cline{2-3}
		& \multirow{2}{*}{Monotonicity} & Increasing: LenScreenName, NumFollowings, TweetLinkRatio, TweetUniqLinkRatio \\
		& & Decreasing: AgeDays, NumFollowers \\ \cline{2-3}
		& Stability & Low-cost features are stable. Stable constant = $8$. \\ \cline{2-3}
		& High Confidence & $\delta = 0.98$. Attacker is allowed to perturb any one of the low-cost features, but not multiple ones. \\ \cline{2-3}
		& \multirow{2}{*}{Redundancy} & $\delta = 0.98, M=2$, any 2 in 4 groups satisfy redundancy: 1) LenScreenName ($\geq 5$ char), LenProfileDescription \\
		& & 2) NumTweets, NumDailyTweets 3) TweetLinkRatio, TweetUniqLinkRatio 4) TweetAtRatio, TweetUniqAtRatio \\ \cline{2-3}
		& Small Neighborhood & $\epsilon = 0.1, c = 50$ \\ \cline{2-3}
	    & Combined & Monotonicity, stability, high confidence, redundancy, and small neighborhood \\ \hline
		\multirow{6}{*}{\begin{tabular}[c]{@{}c@{}}Twitter Spam \\ URLs \end{tabular}}
		& - & Low-cost features: Mention Count, Hashtag Count, Tweet Count, URL Percent. \\ \cline{2-3}
		& \multirow{2}{*}{Monotonicity} & Increasing: 7 shared resources features. EntryURLid, AvgURLid, ChainWeight, \\
		& & CCsize, MinRCLen, AvgLdURLDom, AvgURLDom \\ \cline{2-3}
		& Stability & Low-cost features are stable. Stable constant = $8$. \\ \cline{2-3}
		& High Confidence & $\delta = 0.98$. Attacker is allowed to perturb any one of the low-cost features, but not multiple ones.\\ \cline{2-3}
		& Small Neighborhood & $\epsilon = 1.5, c = 10$ \\ \hline
	\end{tabular} 
	\caption{Global robustness property specifications for three datasets.}
	\label{tab:specifications}
\end{table*}

\subsection{Datasets and Property Specifications}
\label{sec:Datasets and Property Specifications}

We evaluate how well our training technique works
on three security datasets of different scale: detection of cryptojacking~\cite{kharraz2019outguard},
Twitter spam accounts~\cite{lee2011seven}, and Twitter spam URLs~\cite{kwon2017domain}.
Table~\ref{tab:datasets} shows the size of the datasets.
In total, the three datasets have 4K, 40K, and 422,672 data points
respectively.
Appendix~\ref{appendix:Classification Features} lists
all the features for the three datasets.
We specify global robustness properties
for each dataset (Table~\ref{tab:specifications}) based on our analysis of what kinds of evasion strategies might be relatively easy and inexpensive for attackers to perform.

\noindent\textbf{Monotonic Directions.} To specify monotonicity properties,
we use two types of security domain knowledge, suspiciousness and economic cost.
We specify a classifier to be monotonically increasing for a feature if,
(1) an input is more suspicious as the feature value increases,
or, (2) a feature requires a lot of money to be decreased but
easier to be increased, such that we force the attackers
to spend more money in order to reduce the classification score.
Similarly, we specify a classifier to be monotonically decreasing along
a feature dimension by analyzing these two aspects.

\subsubsection{Cryptojacking}
Crytpojacking websites are malicious webpages that hijack user's computing power to mine cryptocurrencies. 
Kharraz et al.~\cite{kharraz2019outguard} collected cryptojacking website data from 12 families of mining libraries. 
We randomly split the dataset containing 2000 malicious websites and 2000 benign websites into 70\% training set and 30\% testing data. In total, there are 2800 training samples and 1200 testing samples.
We use the training set as the validation set.

\noindent\textbf{Low-cost feature.} Among all features, only the hash function feature
is low cost to change. The attacker may use a hash function not on the list, or may
construct aliases of the hash functions to evade the detection. 
Since the other features are related to usage of standard APIs or essential to running
high performance cryptocurrency mining code, they are not trivial to evade.

\noindent\textbf{Monotonicity.} We specify all features to be monotonically increasing. Kharraz et al.~\cite{kharraz2019outguard} proposed seven features to classify cryptojacking websites. A website is more suspicious if any of these features have larger values.
Specifically, cryptojacking websites prefer to use WebSocket APIs to
reduce network communication bandwidth, use WebAssembly to run mining code faster,
runs parallel mining tasks, and may use a list of hash functions. Also, if
a website uses more web workers, has higher messageloop load,
and PostMessage event load, it is more suspicious are performing some heavy load
operations such as coin mining.

\noindent\textbf{Stability.} Since this is a small dataset, we specify all features to be stable,
with stable constant $0.1$.

\noindent\textbf{High Confidence.} 
We use high confidence threshold $0.98$. 

\noindent\textbf{Small Neighborhood.} We specify $\epsilon = 0.2, c = 0.5$. Each feature is allowed to be perturbed by up to $20\%$ of its standard deviation,
and the output of the classifier is bounded by $0.01$.

\subsubsection{Twitter Spam Accounts}

Lee et al.~\cite{lee2011seven} used social honeypot to collect information about Twitter spam accounts,
and randomly sampled benign Twitter users. We reimplement 15 of their proposed features,
including account age, number of following, number of followers, etc., with the entire list
in Table~\ref{tab:features}, Appendix~\ref{appendix:Classification Features}.
We randomly split the dataset into 36,000 training samples and 4,000 testing samples,
and we use the training set as validation set.

\begin{table*}[ht!]
\begin{minipage}[b]{0.48\textwidth}
\centering
    \small
	\begin{tabular}{c|r|r|r|r}
		\hline
		\textbf{Dataset} & \begin{tabular}[c]{@{}c@{}}\textbf{Training}\\\textbf{set size}\end{tabular} & \begin{tabular}[c]{@{}c@{}}\textbf{Test}\\\textbf{set size}\end{tabular}  & \begin{tabular}[c]{@{}c@{}}\textbf{Validation}\\\textbf{set size}\end{tabular} & \begin{tabular}[c]{@{}c@{}}\textbf{\# of}\\\textbf{features}\end{tabular}\\\hline\hline
		Cryptojacking~\cite{kharraz2019outguard}  & 2800 & 1200 & Train & 7 \\\hline
		\begin{tabular}[c]{@{}c@{}}Twitter Spam \\ Accounts~\cite{lee2011seven} \end{tabular} & 36,000 & 4,000 & Train & 15 \\\hline
		\begin{tabular}[c]{@{}c@{}}Twitter Spam \\ URLs~\cite{kwon2017domain} \end{tabular} & 295,870 & 63,401 & 63,401 & 25 \\\hline
	\end{tabular} 
	\caption{The three datasets we use to evaluate our methods.
	For cryptojacking and Twitter spam account datasets, we use the training set as the validation set.}
	\label{tab:datasets}
\end{minipage}\hfill
\begin{minipage}[b]{0.24\textwidth}
\centering
	\begin{tabular}{c|c}
		\hline
		\textbf{\# Char} & \textbf{Price}  \\\hline\hline
		$\leq 4$ & \$1,598.09  \\\hline
		$\geq 5$ & \$298.40 \\ \hline
		Unspecified & \$147.62 \\ \hline
	\end{tabular} 
	\caption{Average price of for-sale Twitter accounts with different number of characters for the username.}
	\label{tab:charprice}
\end{minipage}\hfill
\begin{minipage}[b]{0.23\textwidth}
\includegraphics[width=0.98\textwidth]{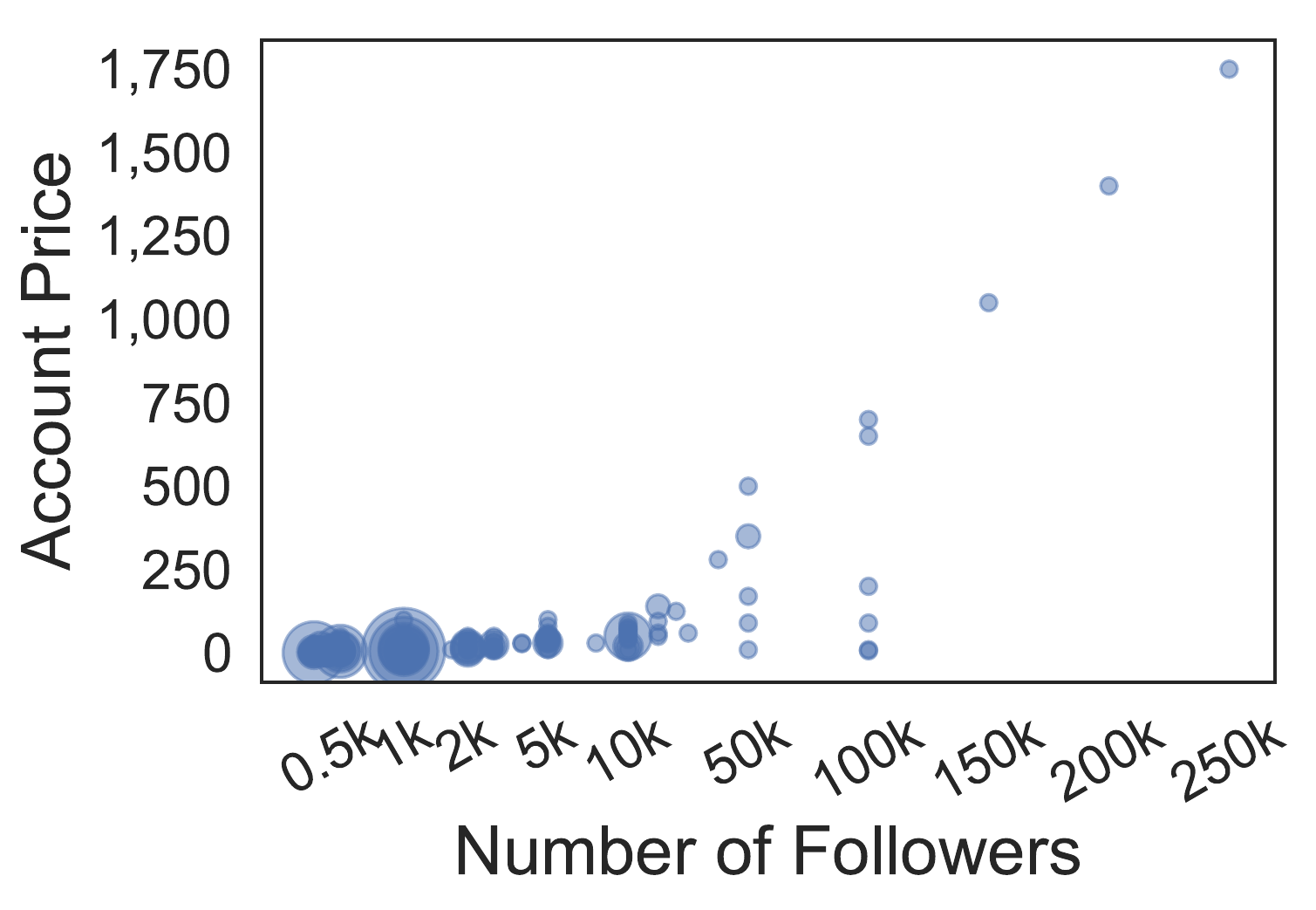}
\captionof{figure}{Price (\$) of for-sale Twitter accounts with different number of followers.}
\label{fig:followerprice}
\end{minipage}\hfill
\end{table*}

\noindent\textbf{Economic Cost Measurement Study.}
We have crawled and analyzed 6,125 for-sale Twitter account posts from an underground forum to measure the effect of LenScreenName and NumFollowers
on the prices of the accounts.
\begin{itemize}[leftmargin=*]
\item \textbf{LenScreenName.}
Accounts with at most 4 characters are deemed special in the underground forum, usually on sale with a special tag `3-4 Characters'.
Table~\ref{tab:charprice} shows that the average price of accounts
with at most 4 characters is five times the price of accounts with more characters
or unspecified characters. More measurement results are in Appendix~\ref{appendix:LenScreenName Feature for Twitter Spam Account Dataset}.
\item \textbf{NumFollowers.}
We measure the account price distribution according to different tiers of followers
indicated in the underground forum, from 500, 1K, 2K up to 250K followers.  
As shown in Figure~\ref{fig:followerprice}, the account prices increase
as the number of followers increases.
\end{itemize}

\noindent\textbf{Low-cost Features.} We identify 8 low-cost features in total. Among them, two features are related to
the user profile, LenScreenName and LenProfileDescription. According to our economic
cost measurement study, accounts with user names up to 4 characters are considered
high cost to obtain. Therefore, we specify LenScreenName with at least 5 characters to be low cost
feature range. The other four low-cost features are related to the tweet content, since they
can be trivially modified by the attacker: NumTweets, NumDailyTweets, TweetLinkRatio, TweetUniqLinkRatio, TweetAtRatio, and TweetUniqAtRatio.

\noindent\textbf{Monotonicity.} We specify two features to be monotonically increasing,
and two features to be monotonically decreasing, based on domain knowledge about
suspicious behavior and economic cost measurement studies.

\emph{Increase in suspiciousness:} Spammers tend to follow a lot of people, expecting social reciprocity
to gain followers for spam content, so large NumFollowings makes an account more suspicious.
If an account sends a lot of links (TweetLinkRatio and TweetUniqLinkRatio), it also becomes more suspicious. 

\emph{Decrease in suspiciousness:} Since cybercriminals are constantly trying to evade blocklists, if an account is newly registered with a small AgeDays value, it is more suspicious.

\emph{Increase in economic cost:}
Since the attacker needs to spend more money to obtain Twitter accounts
with very few characters,
we specify the LenScreenName to be monotonically increasing,

\emph{Decrease in economic cost:}
Since it is expensive for attackers to obtain more followers, we specify the NumFollowers feature to be monotonically decreasing.

\noindent\textbf{Stability.} We specify all the low-cost features to be stable, with stable constant 8. 

\noindent\textbf{High Confidence.}
We allow the attacker to modify any one of the low cost features individually,
but not together.
We use a high confidence prediction threshold $0.98$. 

\noindent\textbf{Redundancy.} Among the 8 low-cost features, we identify four groups,
where each group has one feature that counts an item in total, and one other feature that counts
the same item in a different granularity (daily or unique count).
We specify that any two groups are redundancy of each other ($M=2$)
with $\delta=0.98.$

\noindent\textbf{Small Neighborhood.} We specify $\epsilon = 0.1, c = 50$. The attacker
can change each feature up to $10\%$ of its standard deviation value, and the classifier output change
is bounded by $5$.



\begin{table*}[ht!]
  \centering
  \small
  \begin{tabular}{l | rrrrr | ccccc}
    \hline
    & \multicolumn{5}{|c}{\textbf{Performance}} & \multicolumn{5}{|c}{\textbf{Global Robustness Properties}} \\
    \textbf{Model} & TPR & FPR & Acc & \multirow{2}{*}{AUC} & \multirow{2}{*}{F1} & \multirow{2}{*}{Monotonicity} & \multirow{2}{*}{Stability} & \multirow{2}{*}{\begin{tabular}{@{}c@{}}{High}\\{Confidence}\end{tabular}} & \multirow{2}{*}{Redundancy} & \multirow{2}{*}{\begin{tabular}{@{}c@{}}{Small}\\{Neighborhood}\end{tabular}} \\
    & (\%) & (\%) & (\%) &  &  & & & & & \\
    \hline \hline
    \multicolumn{11}{l}{\textbf{Cryptojacking Detection}} \\ \hline \hline
    XGB & 100 & 0.3 & 99.8 & .99917 & .998 & \xmark & \xmark & \cmark & N/A & \xmark \\ \hline
    Neural Network & 100 & 0.2 & 99.9 & .99997 & .999 & \xmark & \xmark & ? & N/A & \xmark  \\ \hline
    \multicolumn{11}{l}{\textbf{Models with Monotonicity Property}} \\ \hline
    Monotonic XGB & 99.8 & 0.3 & 99.8 & .99969 & .998 & \cmark & \xmark & \cmark & N/A & \xmark \\ \hline
    Nonnegative Linear & 97.7 & 0.2 & 98.8 & .99987 & .988 & \cmark & \xmark & \xmark & N/A & \xmark \\ \hline
    Nonnegative Neural Network & 99.7 & 0.2 & 99.8 & .99999 & .998 & \cmark & \xmark & ? & N/A & \xmark \\ \hline
    Generalized UMNN & 99.8 & 0.2 & 99.8 & .99998 & .998 & \cmark & \xmark & ? & N/A & \xmark \\ \hline
    \multicolumn{11}{l}{\textbf{DL2 Models with Local Robustness Properties, trained using PGD attacks}} \\ \hline
    DL2 Monotoncity & 99.7 & 0.2 & 99.8 & .99999 & .998 & \xmark & \xmark & ? & N/A & \xmark \\ \hline
    DL2 Stability & 99.8 & 0.8 & 99.5 & .99987 & .995 & \xmark & \xmark & \xmark & N/A & \xmark \\ \hline
    DL2 High Confidence & 99.7 & 0.2 & 99.8 & .99999 & .998 & \xmark & \xmark & ? & N/A & \xmark \\ \hline
    DL2 Small Neighborhood & 99.8 & 0.3 & 99.8 & .99999 & .998 & \xmark & \xmark & ? & N/A & \xmark  \\ \hline
    DL2 Combined & 99.3 & 0.2 & 99.6 & .99985 & .996 & \xmark & \xmark & \xmark & N/A & \xmark   \\ \hline
    \multicolumn{11}{l}{\textbf{Our Models with Global Robustness Properties}} \\ \hline
    Logic Ensemble Monotoncity & 100 & 0.3 & 99.8 & .99999 & .998 & \cmark & \xmark  & \cmark & N/A & \xmark  \\ \hline
    Logic Ensemble Stability & 100 & 0.3 & 99.8 & .99831 & .998 & \xmark & \cmark & \cmark & N/A & \cmark \\ \hline
    Logic Ensemble High Confidence & 100 & 0.3 & 99.8 & .99980 & .998 & \xmark & \xmark & \cmark & N/A & \xmark  \\ \hline
    Logic Ensemble Small Neighborhood & 100 & 0.3 & 99.8 & .99961 & .998 & \xmark & \cmark & \cmark & N/A & \cmark  \\ \hline
    Logic Ensemble Combined & 100 & 3.2 & 98.4 & .99831 & .985 & \cmark & \cmark & \cmark & N/A & \cmark  \\ \hline
  \end{tabular}
  \caption{Results for training cryptojacking classifier with global robustness properties, compared to baseline models. N/A: property not specified. $\cmark$: verified to satisfy the property. $\xmark$: verified to not satisfy the property. $?$: unknown.
  }
  \label{tab:cryptojacking_results}
\end{table*}

\subsubsection{Twitter Spam URLs}
Kwon et al.~\cite{kwon2017domain} crawled 15,828,532 tweets by 1,080,466 users. They proposed to use URL redirection chains and
and graph related features to classify spam URL posted on Twitter.
We obtain their public dataset and re-extract
25 features according to the description in the paper.
We extract four categories of features.
(1) Shared resources features capture that the attacker reuse resources
such as hosting servers and redirectors. (2) Heterogeneity-driven features
reflect that attack resources may be heterogeneous and
located around the world. (3) Flexibility-driven features capture
that attackers use different domains and initial URLs
to evade blocklists. (4) Tweet content features measure
the number of special characters, tweets, percentage of URLs
made by the same user.
This is the largest dataset in our evaluation, containing 422,672 samples in total.
We randomly split the dataset into $70\%$ training, $15\%$ testing, and
$15\%$ validation sets.


\noindent\textbf{Low-cost Features.} We specify four tweet content related features to be low cost, since the attacker can trivially modify the content.
They are, Mention Count, Hashtag Count, Tweet Count, and URL percent in tweets.
All the other features are high cost, since they are related to the graph of redirection chains, which
cannot be easily controlled by the attacker. Redirection chains form the
traffic distribution systems in the underground economy, where
different cybercriminals can purchase and re-sell the traffic~\cite{li2013finding,tds}. Thus
graph-related features are largely outside the control of a single attacker,
and are not trivial to change.

\noindent\textbf{Monotonicity.} Based on feature distribution measurement result,
we specify that 7 shared resources-driven features are monotonically increasing,
as shown in Table~\ref{tab:specifications}.
Example measurement result is in Appendix~\ref{appendix:CCSize Feature for Twitter Spam URL Dataset}.

\noindent\textbf{Stability}. We specify low-cost features to be stable,
with stable constant 8.

\noindent\textbf{High Confidence.} We use a high confidence prediction threshold $0.98$. Attacker is allowed to perturb any one of the low-cost features, but not multiple ones.

\noindent\textbf{Small Neighborhood.} We specify $\epsilon = 1.5, c = 10$, which means that the attacker
can change each feature up to $1.5$ times of its standard deviation, and the classifier output change
is bounded by $15$.

\subsection{Baseline Models}

\subsubsection{Experiment Setup}
\label{sec:Experiment Setup}
We compare against three types of baseline models,
(1) tree ensemble and neural network that are not trained
using any properties,
(2) monotonic classifiers, and (3) neural network models trained with
local robustness versions of our properties.

We train the following monotonic classifiers:
monotonic gradient boosted decision trees using XGBoost (Monotonic XGB),
linear classifier with nonnegative weights trained using logistic loss
(Nonnegative Linear), nonnegative neural network, and
generalized unconstrained monotonic neural network (UMNN)~\cite{wehenkel2019unconstrained}.
To evaluate against models with other properties,
we train local versions of our properties using DL2~\cite{fischer2018dl2},
which uses adversarial training.

\noindent\textbf{Malicious Class Gradient Weight.} Since the Twitter spam account
dataset~\cite{lee2011seven} is missing some important features,
we could not reproduce the exact model performance
stated in the paper. Instead, we get 6\% false positive rate.
We contacted the authors but they don't have
the missing data. Therefore, we tune the weight for the gradient of
the malicious class in order to maintain low false positive rate for the models.
We use line search to find the best weight from 0.1 to 1, which increments
by 0.1. We find that using 0.2 to weigh the gradient of the malicious class
can keep the training false positive rate around 2\% for this dataset. 
For the other two datasets, we do not weigh the gradients for different classes.

\noindent\textbf{Linear Classifier.} The nonnegative linear classifier is
a linear combination of input features with nonnegative weights,
trained using logistic loss.
If a feature is specified to be monotonically decreasing, we weigh
the feature by -1 at input.

\noindent\textbf{XGBoost Models.}
For the XGB model and Monotonic XGB model, we specify the following
hyperparameters for three datasets. We use 4 boosting round, max depth 4
per tree to train the cryptojacking classifier, and 
10 boosting rounds, max depth 5 to train Twitter spam account and
Twitter spam URL classifiers.

\noindent\textbf{Neural Network Models.}
The neural networks without any robustness properties as well as the nonnegative-weights networks have two fully connected layers, each with 200, 500, and 300 ReLU units for Cryptojacking, Twitter spam account, and Twitter Spam URL detection respectively. The generalized UMNNs, on the other hand, are positive linear combinations of multiple UMNN each with two fully connected layers and 50, 100, 100 ReLU nodes for each single monotonic feature.

We also use DL2 to train neural networks as baselines, which can achieve local robustness properties using adversarial training. All the DL2 models share the same architectures as the regular neural networks and the training objectives is to minimize the loss of PGD adversarial attacks~\cite{kurakin2016adversarial} that target the robustness properties. We use 50 iterations with step sizes equal to one sixth of the allowable perturbation ranges for PGD attacks in the training process. For testing, we use the same PGD iterations and step sizes but with 10 random restarts.

For all the baseline neural networks mentioned above, we train 50 epochs to minimize binary cross-entropy loss on training datasets using Adam optimizer with learning rate 0.01 and piecewise learning rate scheduler.


\subsubsection{Global Robustness Property Evaluation}

To evaluate whether the baseline models have obtained
global robustness properties, we use our Integer Linear Programming verifier
to verify the XGB and linear models. For neural network models,
we use PGD attacks to maximize the loss function for the property,
as described in Section~\ref{sec:Experiment Setup}.

Table~\ref{tab:cryptojacking_results} and Table~\ref{tab:twitter_results}
show the results of evaluating global robustness properties for
the baseline models.
For neural network models, if the PGD attack has found counterexample for the property, we
consider that the network does not satisfy the property.
Otherwise, we use ``?'' to mark it as unknown/unverified.

\textbf{Result 1: Monotonic XGB and generalized UMNN have the best true
positive rate (TPR) among monotonic classifiers.} For the two relatively large
datasets, the performance of monotonic XGB and generalized UMNN are much
better than nonnegative-weights models. For Twitter spam account detection,
the TPR of monotonic XGB is $16.6\%$ higher than the nonnegative linear classifier.

\textbf{Result 2: Some baseline models naturally satisfy a few
global robustness properties.} The monotonic XGB model for crytpojacking
detection satisfies the high confidence property, because it does not
use the low-cost feature ``hash function'' in the tree structure.
In comparison, our technique can train logic ensemble classifiers to satisfy the high confidence property
but still use the low-cost feature to improve accuracy.
Also, the nonnegative linear classifier for Twitter spam account detection satisfies the small neighborhood property, but it has only $70.1\%$ TPR.
Linear classifiers are known to be robust against small changes in input,
however they have poor performance for many datasets.

\textbf{Result 3: DL2 models cannot obtain global robustness.} We found
counterexamples for all DL2 models for Twitter spam URL detection, using PGD attacks over the property constraint loss, and most models trained with cryptojacking and Twitter spam account detection datasets. If the PGD attack fails to find a counterexample, it does not mean that
the model is verified to have the global property.
There are always stronger attacks that may find counterexamples,
as is often observed with adversarially trained models.

\subsection{Robust Logic Ensembles}
\label{sec:Robust Logic Ensembles}

\subsubsection{Training Algorithm Implementation}

We implement our booster-fixer framework as the following.
We use gradient boosting from XGBoost~\cite{chen2016xgboost} as the booster.
Within each round, we use the booster to add one tree
to the existing classifier, and encode the classifier
as the logic ensemble. This gives the fixer
the structure of clauses and weights ($\alpha, \beta, R$ )
as the starting classifier with high accuracy.

To implement the verifier in the fixer, we use APIs from Gurobi~\cite{gurobi}
to encode the integer linear program with boolean variables
and property violation constraints, and then call the Gurobi solver
to verify the global robustness properties of the logic ensemble.
If the solver returns that the interger linear program is infeasible,
the classifier is verified to satisfy the property. Otherwise,
we construct a counterexample according to solutions
for the boolean variables. For the trainer,
we use PyTorch to implement the smoothed classifier
as Continuous Logic Networks~\cite{ryan2019cln2inv, yao2020learning}.
Then, we use quadratic programming to implement
projected gradient descent.
We compute the updated weights by minimizing the $\ell_2$ norm between the initial weights
and the convex set defined by the training constraints.
We implement the mini-batch training for the smoothed classifier,
where we can specify the batch size. After one epoch of
training, we discretize the classifier to the logic ensemble encoding
for the verifier to verify the property again.
We also implement a few heuristics to speed up the
time for the verifier to generate counterexamples,
with details described in Appendix~\ref{appendix:Heuristics}.


\begin{table}[t!]
  \centering
  \begin{tabular}{l | c}
    \hline
    \textbf{Dataset} & \textbf{Median Training Time} \\ \hline \hline
    Cryptojacking & 25 min \\ \hline
    Twitter Spam Account & 29 hours \\ \hline
    Twitter Spam URL & 3 days \\ \hline
  \end{tabular}
  \caption{Median training time for Logic Ensemble models.
  }
  \label{tab:training_overhead}
\end{table}

\begin{table*}[ht!]
  \centering
  \small
  \begin{tabular}{l | rrrrr | ccccc}
    \hline
    & \multicolumn{5}{|c}{\textbf{Performance}} & \multicolumn{5}{|c}{\textbf{Global Robustness Properties}} \\
    \textbf{Model} & TPR & FPR & Acc & \multirow{2}{*}{AUC} & \multirow{2}{*}{F1} & \multirow{2}{*}{Monotonicity} & \multirow{2}{*}{Stability} & \multirow{2}{*}{\begin{tabular}{@{}c@{}}{High}\\{Confidence}\end{tabular}} & \multirow{2}{*}{Redundancy} & \multirow{2}{*}{\begin{tabular}{@{}c@{}}{Small}\\{Neighborhood}\end{tabular}} \\
    & (\%) & (\%) & (\%) &  &  & & & & & \\
    \hline \hline
    \multicolumn{6}{l}{\textbf{Twitter Spam Account Detection}} \\ \hline \hline
    XGB & 87.0 & 2.3 & 92.2 & .98978 & .920 & \xmark & \xmark & \xmark & \xmark & \xmark  \\ \hline
    Neural Network & 86.4 & 2.5 & 91.8 & .98387 & .915 & \xmark & \xmark & \xmark & \xmark & \xmark  \\ \hline
    \multicolumn{11}{l}{\textbf{Models with Monotonicity Property}} \\ \hline
    Monotonic XGB & 86.7 & 2.7 & 91.9 & .98865 & .916 & \cmark & \xmark & \xmark & \xmark & \xmark  \\ \hline
    Nonnegative Linear & 70.1 & 2.4 & 83.5 & .95321 & .814 & \cmark & \xmark & \xmark & \xmark &  \cmark \\ \hline
    Nonnegative Neural Network & 78.3 & 2.5 & 87.6 & .96723 & .867 & \cmark & \xmark & \xmark & \xmark & \xmark  \\ \hline
    Generalized UMNN & 86.0 & 3.9 & 90.9 & .97324 & .907 & \cmark & \xmark & \xmark & \xmark & \xmark  \\ \hline
    \multicolumn{11}{l}{\textbf{DL2 Models with Local Robustness Properties, trained using PGD attacks}} \\ \hline
    DL2 Monotoncity & 83.2 & 2.6 & 90.1 & .97800 & .896 & \xmark & \xmark & \xmark & \xmark & \xmark  \\ \hline
    DL2 Stability & 86.1 & 3.3 & 91.3 & .98029 & .910 & \xmark & ? & \xmark & \xmark & \xmark  \\ \hline
    DL2 High Confidence & 82.8 & 2.6 & 89.9 & .98056 & .894 & \xmark & \xmark & \xmark & \xmark & \xmark  \\ \hline
    DL2 Redundancy & 83.9 & 3.1 & 90.2 & .97898 & .898 & \xmark & \xmark & ? & \xmark & \xmark  \\ \hline
    DL2 Small Neighborhood & 88.3 & 3.5 & 92.2 & .98086 & .921 & \xmark & \xmark & \xmark & \xmark & \xmark \\ \hline
    DL2 Combined & 83.8 & 3 & 90.2 & .97738 & .898 & \xmark & \xmark & \xmark & \xmark & ?  \\ \hline
    \multicolumn{11}{l}{\textbf{Our Models with Global Robustness Properties}} \\ \hline
    Logic Ensemble Monotoncity & 83.2 & 3.2 & 89.8 & .97297 & .894 & \cmark & \xmark & \xmark & \xmark & \xmark  \\ \hline
    Logic Ensemble Stability & 86.0 & 2.1 & 91.8 & .98479 & .915 & \xmark & \cmark & \xmark & \xmark &  \xmark \\ \hline
    Logic Ensemble High Confidence & 86.1 & 2.6 & 91.6 & .98311 & .913 & \xmark & \cmark & \cmark & \xmark & \xmark  \\ \hline
    Logic Ensemble Redundancy & 85.5 & 3.2 & 91.0 & .98166 & .907 & \xmark & \cmark & \cmark & \cmark & \xmark \\ \hline
    Logic Ensemble Small Neighborhood & 83.9 & 2.5 & 90.5 &  .98325 & 0.901 & \xmark & \cmark & \xmark & \xmark & \cmark  \\ \hline
    Logic Ensemble Combined & 81.6 & 2.4 & 89.4 & .98142 & .888 & \cmark & \cmark & \cmark & \cmark & \cmark  \\ \hline \hline
    \multicolumn{6}{l}{\textbf{Twitter Spam URL Detection}} \\ \hline \hline
    XGB & 99.0 & 1.5 & 98.7 & .99834 & .986 & \xmark & \xmark & \xmark & N/A & \xmark  \\ \hline
    Neural Network & 98.8 & 2.9 & 97.9 & .99735 & .977 & \xmark & \xmark & \xmark & N/A & \xmark \\ \hline
    \multicolumn{11}{l}{\textbf{Models with Monotonicity Property}} \\ \hline
    Monotonic XGB & 99.4 & 1.7 & 98.8 & .99848 & .986 & \cmark & \xmark & \xmark & N/A & \xmark  \\ \hline
    Nonnegative Linear & 93.2 & 18.6 & 86.7 & .90218 & .861 & \cmark & \xmark & \xmark & N/A & \xmark  \\ \hline
    Nonnegative Neural Network & 98.0 & 6.9 & 95.3 & .98511 & .949 & \cmark & \xmark & \xmark & N/A & \xmark  \\ \hline
    Generalized UMNN & 98.8 & 2.6 & 98.0 & .99732 & .977 & \cmark & \xmark & \xmark & N/A & \xmark  \\ \hline
    \multicolumn{11}{l}{\textbf{DL2 Models with Local Robustness Properties, trained using PGD attacks}} \\ \hline
    DL2 Monotoncity & 98.9 & 3.0 & 97.9 & .99694 & .976 & \xmark & \xmark & \xmark & N/A & \xmark  \\ \hline
    DL2 Stability & 99.0 & 3.0 & 97.9 & .99706 & .977 & \xmark & \xmark & \xmark & N/A &  \xmark \\ \hline
    DL2 High Confidence & 99.5 & 4.6 & 97.2 & .99696 & .969 & \xmark & \xmark & \xmark & N/A &  \xmark \\ \hline
    DL2 Small Neighborhood & 99.1 & 3.0 & 97.9 & .99720 & .977 & \xmark & \xmark & \xmark & N/A & \xmark  \\ \hline
    \multicolumn{11}{l}{\textbf{Our Models with Global Robustness Properties}} \\ \hline
    Logic Ensemble Monotoncity & 96.3 & 3.5 & 96.4 & .98549 & .960 & \cmark & \xmark & \xmark & N/A & \xmark  \\ \hline
    Logic Ensemble Stability & 92.9 & 3.3 & 95.0 & .98180 & .943 & \xmark & \cmark & \xmark & N/A &  \cmark \\ \hline
    Logic Ensemble High Confidence & 97.6 & 5.4 & 95.9 & .98646 & .955 & \xmark & \cmark & \cmark & N/A & \xmark  \\ \hline
    Logic Ensemble Small Neighborhood & 97.1 & 2.8 & 97.1 & .99338 & .968 & \xmark & \xmark & \xmark & N/A & \cmark  \\ \hline
  \end{tabular}
  \caption{Results for training Twitter account classifier and Twitter spam URL classifier with global robustness properties, compared to baseline models. N/A: property not specified. $\cmark$: verified to satisfy the property. $\xmark$: verified to not satisfy the property. $?$: unknown.
  }
  \label{tab:twitter_results}
\end{table*}

\subsubsection{Experiment Setup.}
For the cryptojacking dataset, we boost 4 rounds, each adding a tree with
max depth 4. For the other two datasets, we boost 10 rounds, with max tree depth 5, except that we only boost 6 rounds when training the Twitter spam account classifier with all five properties.
During CLN training, we keep track of the discrete classifier at each stage, including all the inequalities and conjunctions. When we need to smooth the classifier, we use shifted and scaled sigmoid function to smooth the inequality, with temperature $\frac{1}{500}$, shift by $0.01$, and product t-norm to smooth the conjunctions, to closely approximate the discrete classifier. The updated weights from gradient-guided training can be directly used for the discrete classifier. To discretize the model, we simply do not apply the sigmoid function and the product t-norm.
We use the Adam optimizer with learning rate $0.001$
and decay $0.95$, to minimize binary cross-entropy loss using gradient descent.
For the crytpojacking dataset, we use mini-batch size 1; for
the other two larger datasets, we use mini-batch size 1024.
After boosting all the rounds, we choose the model with the highest validation AUC.
Empirically, our algorithm converges well to an accurate classifier that satisfies
the specified properties.

\subsubsection{Global Robustness Property Evaluation}
\label{sec:Global Robustness Property Evaluation}
We train 15 logic ensemble models in total for the three datasets,
each satisfying the specified global robustness properties, shown in Table~\ref{tab:cryptojacking_results} and Table~\ref{tab:twitter_results}.
We use our Integer Linear Program verifier (Section~\ref{sec:Integer Linear Program Verifier})
to verify the properties for all models.

\textbf{Training Overhead.} Similar to most existing robust machine learning training strategies, training a verifiably robust model is significantly slower than training a non-robust model. We show the median training time for Logic Ensemble models in Table~\ref{tab:training_overhead}. Training non-robust XGBoost models takes one minute. However, computation is usually cheap and the tradeoff for getting more robustness in exchange for more computation is common across robust machine learning techniques. Next, we discuss our key results.

\textbf{Result 4: Our monotonic models have comparable or better performance than
existing methods.} Our Logic Ensemble Monotonicity models have higher true positive rate and AUC than the Nonnegative Linear classifiers for all three datasets, and we also achieve better performance than
the Nonnegative Neural Network models
for the cryptojacking detection and Twitter account detection datasets.
Monotonic XGB outperforms our Logic Ensemble Monotonic models, but we still
have comparable performance. For example, for the Twitter spam account detection,
our Logic Ensemble Monotonicity model has $3.5\%$ lower true positive rate (TPR),
and $0.5\%$ higher false positive rate (FPR) than the Monotonic XGB model.

\textbf{Result 5: Our models have moderate performance drop
to obtain an individual property.}
For cryptojacking detection, enforcing each property does not
decrease TPR at all, and only increases FPR by $0.1\%$ compared to
the baseline neural network model (Table~\ref{tab:cryptojacking_results}).
For Twitter spam account detection,
logic ensemble models that satisfy one global robustness property decrease
the TPR by at most $3.8\%$, and increase the FPR by at most $0.9\%$,
compared to the baseline XGB model (Table~\ref{tab:twitter_results}).
For Twitter spam URL detection, within monotonicity, stability,
and small neighborhood properties, enforcing one property
for the classifier can maintain high TPR (from $92.9\%$ to $97.6\%$)
and low FPR (from $2.8\%$ to $5.4\%$, Table~\ref{tab:twitter_results}).
For example, the Logic Ensemble High Confidence model
decreases the TPR by $1.4\%$ and increases the FPR by $3.9\%$,
compared to the baseline XGB model.
This model utilizes the low-cost features
to improve the prediction accuracy. If we only use high-cost features
to train a tree ensemble with the same capacity (10 rounds of boosting),
we can only achieve $79.9\%$ TPR and $0.96075$ AUC. In comparison,
our Logic Ensemble High Confidence model has $97.6\%$ TPR and $0.98646$ AUC.
Results regarding hyperparameters are discussed in Appendix~\ref{appendix:Hyperparameters}.

\textbf{Result 6: Training a classifier with one property sometimes obtains another property.}
Table~\ref{tab:cryptojacking_results} shows that all cryptojacking
Logic Ensemble classifiers that were enforced with only one property,
have obtained at least one other property. For example,
the Logic Ensemble Stability model
has obtained small neighborhood property, and vice versa.
Since we specify all features to be stable for this dataset,
the stability property is equivalent to
the global Lipschitz property under $L_0$ distance. On the other hand,
we define the small neighborhood property with a new distance.
This shows that enforcing robustness for one property can generalize
the robustness to a different property.
More results are discussed in Appendix~\ref{appendix:Obtaining More Properties}.

\textbf{Result 7: We can train classifiers to satisfy multiple global robustness properties at the same time.}
We train a cryptojacking classifier with four properties, and a Twitter spam account classifier
with five properties. For cryptojacking detection, the Logic Ensemble Combined
model maintains the same high TPR, and only increases the FPR by $3\%$
compared to the baseline neural network model (Table~\ref{tab:cryptojacking_results}).
For Twitter spam account detection, the Logic Ensemble Combined model that satisfies all properties
only decreases the TPR by \socialtprdecrease{}
and increases the FPR by \socialfprincrease{},
compared to the baseline XGB model with no property (Table~\ref{tab:twitter_results}).
More results are discussed in Appendix~\ref{appendix:Logic Ensemble Combined Model}.

\section{Related Work}


\noindent\textbf{Program Synthesis.}
Solar-Lezama et al.~\cite{solar2006combinatorial} proposed counterexample guided inductive synthesis (CEGIS)
to synthesize finite programs according to specifications of desired functionalities.
The key idea is to iteratively generate a proposal of the program
and check the correctness of the program, where the checker should be able
to generate counterexamples of correctness to guide the program generation process.
The general idea of CEGIS has also been used to learn
recursive logic programs (e.g., as static analysis rules)~\cite{cropper2020turning,albarghouthi2017constraint,si2018syntax,raghothaman2019provenance}.
We design our fixer following the general form of CEGIS.




\noindent\textbf{Local Robustness.}
Many techniques have been proposed to verify local robustness (e.g., $\ell_p$ robustness) of neural networks,
including customized solvers~\cite{katz2017reluplex,tjeng2017evaluating,huang2017safety,katz2019marabou}
and bound propagation based verification methods~\cite{li2020sok,wong2018provable,raghunathan2018semidefinite,reluval2018,shiqi2018efficient,wang2021beta,xu2020fast,singh2019abstract,singh2019boosting,balunovic2019certifying,singh2019beyond,muller2021precise,muller2021scaling,zhang2018efficient,boopathy2019cnn,weng2018towards,salman2019convex}. Bound propagation verifiers can also be applied in robust optimization to train the models with certified local robustness~\cite{wong2018scaling,zhang2020towards,mirman2018differentiable,wang2018mixtrain,boopathy2021fast,li2019robustra,zhang2018cost,chen2019training}. Randomized smoothing~\cite{cohen2019certified,jia2019certified,lecuyer2019certified,salman2019provably,yang2020randomized,li2018certified} is another technique to provide probabilistic local robustness guarantee.
Several methods have been proposed to utilize the local Lipshitz constant of neural networks
for verification~\cite{weng2018towards,weng2018evaluating,hein2017formal}, and constrain or use the local Lipshitz bounds to train robust networks~\cite{szegedy2013intriguing,cisse2017parseval,cohen2019universal,anil2019sorting,pauli2021training,qin2019adversarial,finlay2021scaleable,lee2020lipschitz,gouk2021regularisation,singla2019bounding,farnia2018generalizable}.




\noindent\textbf{Global Robustness.}
Fischer et al.~\cite{fischer2018dl2} and Melacci et al.~\cite{melacci2020can} proposed global robustness properties for image classifiers using universally quantified statements.
Both of their techniques smooth the logic expression of the property into a differentiable loss function,
and then use PGD attacks~\cite{kurakin2016adversarial} to minimize the loss.
They can train neural networks to obtain local robustness,
but cannot obtain verified global robustness.
ART~\cite{lin2020art} proposed an abstraction refinement strategy to train provably correct neural networks.
The model satisfies global robustness properties when the correctness loss reaches zero.
However, in practice their correctness loss did not converge to zero.
Leino et al.~\cite{leino2021globally} proposed to minimize global Lipschitz constant to train globally-robust neural networks, but they can only verify one global property that abstains on non-robust predictions.


\noindent\textbf{Monotonic Classifiers.}
Many methods have been proposed to train monotonic classifiers~\cite{wehenkel2019unconstrained,incer2018adversarially,archer1993application,daniels2010monotone,kay2000estimating,ben1995monotonicity,duivesteijn2008nearest,feelders2010monotone,gupta2016monotonic}.
Recently, Wehenkel et al.~\cite{wehenkel2019unconstrained} proposed unconstrained monotonic neural networks,
based on the key idea that a function is monotonic as long as its derivative is nonnegative.
This has increased the performance of monotonic neural network significantly compared to enforcing nonnegative weights.
Incer et al.~\cite{incer2018adversarially} used monotone constraints from XGBoost
to train monotonic malware classifiers. XGBoost enforces monotone constraints
for the left child weight to be always smaller (or greater) than the right child,
which is a specialized method and does not generalize to other global robustness properties.

\noindent\textbf{Discrete Classifier and Smoothing.}
Friedman et al.~\cite{friedman2008predictive} proposed rule ensemble,
where the each rule is a path in the decision tree, and they
used regression to learn how to combine rules.
Our logic ensemble is more general such that the clauses
do not have to form a tree structure. We only take rules from trees
as the starting classifier to fix the properties.
Kantchelian et al.~\cite{kantchelian2016evasion} proposed the mixed integer linear program
attack to evade tree ensembles by perturbing a concrete input. In comparison,
our integer linear program verifier has only integer variables, and represents
all inputs symbolically. Continuous Logic Networks was proposed
to smooth SMT formulas to learn loop invariants~\cite{ryan2019cln2inv, yao2020learning}.
In this paper, we apply the smoothing techniques to
train machine learning classifiers.



\section{Conclusion}

In this paper, we have presented a novel booster-fixer training framework
to enforce new global robustness properties for security classifiers.
We have formally defined six global robustness properties,
of which five are new. Our training technique is general, and can 
handle a large class of properties. We have used experiments to show that
we can train different security classifiers to satisfy multiple
global robustness properties at the same time.

\section{Acknowledgements}

We thank the anonymous reviewers for their constructive and valuable feedback. This work is supported in part by NSF grants CNS-18-50725, CCF-21-24225; generous gifts from Open Philanthropy, two Google Faculty Fellowships, Berkeley Artificial Intelligence Research (BAIR), a Capital One Research Grant, a J.P. Morgan Faculty Award; and Institute of Information \& communications Technology Planning \& Evaluation (IITP) grant funded by the Korea government(MSIT) (No.2020-0-00153). Any opinions, findings, conclusions, or recommendations expressed herein are those of the authors, and do not necessarily reflect those of the US Government, NSF, Google, Capital One, J.P. Morgan, or the Korea government.




\balance
\bibliographystyle{ACM-Reference-Format}
\bibliography{ref}


\begin{thebibliography}{99}


\ifx \showCODEN    \undefined \def \showCODEN     #1{\unskip}     \fi
\ifx \showDOI      \undefined \def \showDOI       #1{#1}\fi
\ifx \showISBNx    \undefined \def \showISBNx     #1{\unskip}     \fi
\ifx \showISBNxiii \undefined \def \showISBNxiii  #1{\unskip}     \fi
\ifx \showISSN     \undefined \def \showISSN      #1{\unskip}     \fi
\ifx \showLCCN     \undefined \def \showLCCN      #1{\unskip}     \fi
\ifx \shownote     \undefined \def \shownote      #1{#1}          \fi
\ifx \showarticletitle \undefined \def \showarticletitle #1{#1}   \fi
\ifx \showURL      \undefined \def \showURL       {\relax}        \fi
\providecommand\bibfield[2]{#2}
\providecommand\bibinfo[2]{#2}
\providecommand\natexlab[1]{#1}
\providecommand\showeprint[2][]{arXiv:#2}

\bibitem[\protect\citeauthoryear{??}{gur}{[n.d.]}]%
        {gurobi}
 \bibinfo{year}{[n.d.]}\natexlab{}.
\newblock \bibinfo{title}{{Gurobi Optimization}}.
\newblock \bibinfo{howpublished}{\url{https://www.gurobi.com/}}.
\newblock


\bibitem[\protect\citeauthoryear{Albarghouthi, Koutris, Naik, and
  Smith}{Albarghouthi et~al\mbox{.}}{2017}]%
        {albarghouthi2017constraint}
\bibfield{author}{\bibinfo{person}{Aws Albarghouthi},
  \bibinfo{person}{Paraschos Koutris}, \bibinfo{person}{Mayur Naik}, {and}
  \bibinfo{person}{Calvin Smith}.} \bibinfo{year}{2017}\natexlab{}.
\newblock \showarticletitle{Constraint-based synthesis of Datalog programs}. In
  \bibinfo{booktitle}{\emph{International Conference on Principles and Practice
  of Constraint Programming}}. Springer, \bibinfo{pages}{689--706}.
\newblock


\bibitem[\protect\citeauthoryear{Allix, Bissyand{\'e}, Klein, and
  Le~Traon}{Allix et~al\mbox{.}}{2015}]%
        {allix2015your}
\bibfield{author}{\bibinfo{person}{Kevin Allix},
  \bibinfo{person}{Tegawend{\'e}~F Bissyand{\'e}}, \bibinfo{person}{Jacques
  Klein}, {and} \bibinfo{person}{Yves Le~Traon}.}
  \bibinfo{year}{2015}\natexlab{}.
\newblock \showarticletitle{Are your training datasets yet relevant?}. In
  \bibinfo{booktitle}{\emph{International Symposium on Engineering Secure
  Software and Systems}}. Springer, \bibinfo{pages}{51--67}.
\newblock


\bibitem[\protect\citeauthoryear{Anil, Lucas, and Grosse}{Anil
  et~al\mbox{.}}{2019}]%
        {anil2019sorting}
\bibfield{author}{\bibinfo{person}{Cem Anil}, \bibinfo{person}{James Lucas},
  {and} \bibinfo{person}{Roger Grosse}.} \bibinfo{year}{2019}\natexlab{}.
\newblock \showarticletitle{Sorting out Lipschitz function approximation}. In
  \bibinfo{booktitle}{\emph{International Conference on Machine Learning}}.
  PMLR, \bibinfo{pages}{291--301}.
\newblock


\bibitem[\protect\citeauthoryear{Archer and Wang}{Archer and Wang}{1993}]%
        {archer1993application}
\bibfield{author}{\bibinfo{person}{Norman~P Archer} {and}
  \bibinfo{person}{Shouhong Wang}.} \bibinfo{year}{1993}\natexlab{}.
\newblock \showarticletitle{Application of the back propagation neural network
  algorithm with monotonicity constraints for two-group classification
  problems}.
\newblock \bibinfo{journal}{\emph{Decision Sciences}} \bibinfo{volume}{24},
  \bibinfo{number}{1} (\bibinfo{year}{1993}), \bibinfo{pages}{60--75}.
\newblock


\bibitem[\protect\citeauthoryear{Balunovi{\'c}, Baader, Singh, Gehr, and
  Vechev}{Balunovi{\'c} et~al\mbox{.}}{2019}]%
        {balunovic2019certifying}
\bibfield{author}{\bibinfo{person}{Mislav Balunovi{\'c}},
  \bibinfo{person}{Maximilian Baader}, \bibinfo{person}{Gagandeep Singh},
  \bibinfo{person}{Timon Gehr}, {and} \bibinfo{person}{Martin Vechev}.}
  \bibinfo{year}{2019}\natexlab{}.
\newblock \showarticletitle{Certifying geometric robustness of neural
  networks}.
\newblock \bibinfo{journal}{\emph{Advances in Neural Information Processing
  Systems (NeurIPS)}} (\bibinfo{year}{2019}).
\newblock


\bibitem[\protect\citeauthoryear{Ben-David}{Ben-David}{1995}]%
        {ben1995monotonicity}
\bibfield{author}{\bibinfo{person}{Arie Ben-David}.}
  \bibinfo{year}{1995}\natexlab{}.
\newblock \showarticletitle{Monotonicity maintenance in information-theoretic
  machine learning algorithms}.
\newblock \bibinfo{journal}{\emph{Machine Learning}} \bibinfo{volume}{19},
  \bibinfo{number}{1} (\bibinfo{year}{1995}), \bibinfo{pages}{29--43}.
\newblock


\bibitem[\protect\citeauthoryear{Boopathy, Weng, Chen, Liu, and
  Daniel}{Boopathy et~al\mbox{.}}{2019}]%
        {boopathy2019cnn}
\bibfield{author}{\bibinfo{person}{Akhilan Boopathy}, \bibinfo{person}{Tsui-Wei
  Weng}, \bibinfo{person}{Pin-Yu Chen}, \bibinfo{person}{Sijia Liu}, {and}
  \bibinfo{person}{Luca Daniel}.} \bibinfo{year}{2019}\natexlab{}.
\newblock \showarticletitle{Cnn-cert: An efficient framework for certifying
  robustness of convolutional neural networks}. In
  \bibinfo{booktitle}{\emph{AAAI Conference on Artificial Intelligence
  (AAAI)}}.
\newblock


\bibitem[\protect\citeauthoryear{Boopathy, Weng, Liu, Chen, Zhang, and
  Daniel}{Boopathy et~al\mbox{.}}{2021}]%
        {boopathy2021fast}
\bibfield{author}{\bibinfo{person}{Akhilan Boopathy}, \bibinfo{person}{Tsui-Wei
  Weng}, \bibinfo{person}{Sijia Liu}, \bibinfo{person}{Pin-Yu Chen},
  \bibinfo{person}{Gaoyuan Zhang}, {and} \bibinfo{person}{Luca Daniel}.}
  \bibinfo{year}{2021}\natexlab{}.
\newblock \showarticletitle{Fast Training of Provably Robust Neural Networks by
  SingleProp}.
\newblock \bibinfo{journal}{\emph{AAAI Conference on Artificial Intelligence
  (AAAI)}} (\bibinfo{year}{2021}).
\newblock


\bibitem[\protect\citeauthoryear{Chen and Guestrin}{Chen and Guestrin}{2016}]%
        {chen2016xgboost}
\bibfield{author}{\bibinfo{person}{Tianqi Chen} {and} \bibinfo{person}{Carlos
  Guestrin}.} \bibinfo{year}{2016}\natexlab{}.
\newblock \showarticletitle{Xgboost: A scalable tree boosting system}. In
  \bibinfo{booktitle}{\emph{Proceedings of the 22nd acm sigkdd international
  conference on knowledge discovery and data mining}}. ACM,
  \bibinfo{pages}{785--794}.
\newblock


\bibitem[\protect\citeauthoryear{Chen, Wang, She, and Jana}{Chen
  et~al\mbox{.}}{2020}]%
        {chen2019training}
\bibfield{author}{\bibinfo{person}{Yizheng Chen}, \bibinfo{person}{Shiqi Wang},
  \bibinfo{person}{Dongdong She}, {and} \bibinfo{person}{Suman Jana}.}
  \bibinfo{year}{2020}\natexlab{}.
\newblock \showarticletitle{{On Training Robust PDF Malware Classifiers}}. In
  \bibinfo{booktitle}{\emph{USENIX Security Symposium}}.
\newblock


\bibitem[\protect\citeauthoryear{Cisse, Bojanowski, Grave, Dauphin, and
  Usunier}{Cisse et~al\mbox{.}}{2017}]%
        {cisse2017parseval}
\bibfield{author}{\bibinfo{person}{Moustapha Cisse}, \bibinfo{person}{Piotr
  Bojanowski}, \bibinfo{person}{Edouard Grave}, \bibinfo{person}{Yann Dauphin},
  {and} \bibinfo{person}{Nicolas Usunier}.} \bibinfo{year}{2017}\natexlab{}.
\newblock \showarticletitle{Parseval networks: Improving robustness to
  adversarial examples}. In \bibinfo{booktitle}{\emph{International Conference
  on Machine Learning}}. PMLR, \bibinfo{pages}{854--863}.
\newblock


\bibitem[\protect\citeauthoryear{Cohen, Huster, and Cohen}{Cohen
  et~al\mbox{.}}{2019a}]%
        {cohen2019universal}
\bibfield{author}{\bibinfo{person}{Jeremy~EJ Cohen}, \bibinfo{person}{Todd
  Huster}, {and} \bibinfo{person}{Ra Cohen}.} \bibinfo{year}{2019}\natexlab{a}.
\newblock \showarticletitle{Universal lipschitz approximation in bounded depth
  neural networks}.
\newblock \bibinfo{journal}{\emph{arXiv preprint arXiv:1904.04861}}
  (\bibinfo{year}{2019}).
\newblock


\bibitem[\protect\citeauthoryear{Cohen, Rosenfeld, and Kolter}{Cohen
  et~al\mbox{.}}{2019b}]%
        {cohen2019certified}
\bibfield{author}{\bibinfo{person}{Jeremy~M Cohen}, \bibinfo{person}{Elan
  Rosenfeld}, {and} \bibinfo{person}{J~Zico Kolter}.}
  \bibinfo{year}{2019}\natexlab{b}.
\newblock \showarticletitle{Certified adversarial robustness via randomized
  smoothing}.
\newblock \bibinfo{journal}{\emph{International Conference on Machine
  Learning}} (\bibinfo{year}{2019}).
\newblock


\bibitem[\protect\citeauthoryear{Cropper, Duman{\v{c}}i{\'c}, and
  Muggleton}{Cropper et~al\mbox{.}}{2020}]%
        {cropper2020turning}
\bibfield{author}{\bibinfo{person}{Andrew Cropper}, \bibinfo{person}{Sebastijan
  Duman{\v{c}}i{\'c}}, {and} \bibinfo{person}{Stephen~H Muggleton}.}
  \bibinfo{year}{2020}\natexlab{}.
\newblock \showarticletitle{Turning 30: New ideas in inductive logic
  programming}. In \bibinfo{booktitle}{\emph{International Joint Conferences on
  Artifical Intelligence (IJCAI)}}.
\newblock


\bibitem[\protect\citeauthoryear{Daniels and Velikova}{Daniels and
  Velikova}{2010}]%
        {daniels2010monotone}
\bibfield{author}{\bibinfo{person}{Hennie Daniels} {and}
  \bibinfo{person}{Marina Velikova}.} \bibinfo{year}{2010}\natexlab{}.
\newblock \showarticletitle{Monotone and partially monotone neural networks}.
\newblock \bibinfo{journal}{\emph{IEEE Transactions on Neural Networks}}
  \bibinfo{volume}{21}, \bibinfo{number}{6} (\bibinfo{year}{2010}),
  \bibinfo{pages}{906--917}.
\newblock


\bibitem[\protect\citeauthoryear{Duivesteijn and Feelders}{Duivesteijn and
  Feelders}{2008}]%
        {duivesteijn2008nearest}
\bibfield{author}{\bibinfo{person}{Wouter Duivesteijn} {and}
  \bibinfo{person}{Ad Feelders}.} \bibinfo{year}{2008}\natexlab{}.
\newblock \showarticletitle{Nearest neighbour classification with monotonicity
  constraints}. In \bibinfo{booktitle}{\emph{Joint European Conference on
  Machine Learning and Knowledge Discovery in Databases}}. Springer,
  \bibinfo{pages}{301--316}.
\newblock


\bibitem[\protect\citeauthoryear{Dutta, Jha, Sankaranarayanan, and
  Tiwari}{Dutta et~al\mbox{.}}{2018}]%
        {dutta2018output}
\bibfield{author}{\bibinfo{person}{Souradeep Dutta}, \bibinfo{person}{Susmit
  Jha}, \bibinfo{person}{Sriram Sankaranarayanan}, {and}
  \bibinfo{person}{Ashish Tiwari}.} \bibinfo{year}{2018}\natexlab{}.
\newblock \showarticletitle{Output Range Analysis for Deep Feedforward Neural
  Networks}. In \bibinfo{booktitle}{\emph{NASA Formal Methods Symposium}}.
  Springer, \bibinfo{pages}{121--138}.
\newblock


\bibitem[\protect\citeauthoryear{Dvijotham, Gowal, Stanforth, Arandjelovic,
  O'Donoghue, Uesato, and Kohli}{Dvijotham et~al\mbox{.}}{2018a}]%
        {dvijotham2018training}
\bibfield{author}{\bibinfo{person}{Krishnamurthy Dvijotham},
  \bibinfo{person}{Sven Gowal}, \bibinfo{person}{Robert Stanforth},
  \bibinfo{person}{Relja Arandjelovic}, \bibinfo{person}{Brendan O'Donoghue},
  \bibinfo{person}{Jonathan Uesato}, {and} \bibinfo{person}{Pushmeet Kohli}.}
  \bibinfo{year}{2018}\natexlab{a}.
\newblock \showarticletitle{Training verified learners with learned verifiers}.
\newblock \bibinfo{journal}{\emph{arXiv preprint arXiv:1805.10265}}
  (\bibinfo{year}{2018}).
\newblock


\bibitem[\protect\citeauthoryear{Dvijotham, Stanforth, Gowal, Mann, and
  Kohli}{Dvijotham et~al\mbox{.}}{2018b}]%
        {dvijotham2018dual}
\bibfield{author}{\bibinfo{person}{Krishnamurthy Dvijotham},
  \bibinfo{person}{Robert Stanforth}, \bibinfo{person}{Sven Gowal},
  \bibinfo{person}{Timothy Mann}, {and} \bibinfo{person}{Pushmeet Kohli}.}
  \bibinfo{year}{2018}\natexlab{b}.
\newblock \showarticletitle{A dual approach to scalable verification of deep
  networks}.
\newblock \bibinfo{journal}{\emph{arXiv preprint arXiv:1803.06567}}
  (\bibinfo{year}{2018}).
\newblock


\bibitem[\protect\citeauthoryear{Ehlers}{Ehlers}{2017}]%
        {ehlers2017formal}
\bibfield{author}{\bibinfo{person}{Ruediger Ehlers}.}
  \bibinfo{year}{2017}\natexlab{}.
\newblock \showarticletitle{Formal Verification of Piece-Wise Linear
  Feed-Forward Neural Networks}.
\newblock \bibinfo{journal}{\emph{15th International Symposium on Automated
  Technology for Verification and Analysis}} (\bibinfo{year}{2017}).
\newblock


\bibitem[\protect\citeauthoryear{Farnia, Zhang, and Tse}{Farnia
  et~al\mbox{.}}{2018}]%
        {farnia2018generalizable}
\bibfield{author}{\bibinfo{person}{Farzan Farnia}, \bibinfo{person}{Jesse
  Zhang}, {and} \bibinfo{person}{David Tse}.} \bibinfo{year}{2018}\natexlab{}.
\newblock \showarticletitle{Generalizable Adversarial Training via Spectral
  Normalization}. In \bibinfo{booktitle}{\emph{International Conference on
  Learning Representations}}.
\newblock


\bibitem[\protect\citeauthoryear{Feelders}{Feelders}{2010}]%
        {feelders2010monotone}
\bibfield{author}{\bibinfo{person}{Ad Feelders}.}
  \bibinfo{year}{2010}\natexlab{}.
\newblock \showarticletitle{Monotone relabeling in ordinal classification}. In
  \bibinfo{booktitle}{\emph{2010 IEEE International Conference on Data
  Mining}}. IEEE, \bibinfo{pages}{803--808}.
\newblock


\bibitem[\protect\citeauthoryear{Finlay and Oberman}{Finlay and
  Oberman}{2021}]%
        {finlay2021scaleable}
\bibfield{author}{\bibinfo{person}{Chris Finlay} {and} \bibinfo{person}{Adam~M
  Oberman}.} \bibinfo{year}{2021}\natexlab{}.
\newblock \showarticletitle{Scaleable input gradient regularization for
  adversarial robustness}.
\newblock \bibinfo{journal}{\emph{Machine Learning with Applications}}
  \bibinfo{volume}{3} (\bibinfo{year}{2021}), \bibinfo{pages}{100017}.
\newblock


\bibitem[\protect\citeauthoryear{Fischer, Balunovic, Drachsler-Cohen, Gehr,
  Zhang, and Vechev}{Fischer et~al\mbox{.}}{2019}]%
        {fischer2018dl2}
\bibfield{author}{\bibinfo{person}{Marc Fischer}, \bibinfo{person}{Mislav
  Balunovic}, \bibinfo{person}{Dana Drachsler-Cohen}, \bibinfo{person}{Timon
  Gehr}, \bibinfo{person}{Ce Zhang}, {and} \bibinfo{person}{Martin Vechev}.}
  \bibinfo{year}{2019}\natexlab{}.
\newblock \showarticletitle{DL2: Training and Querying Neural Networks with
  Logic}. In \bibinfo{booktitle}{\emph{International Conference on Machine
  Learning (ICML)}}.
\newblock


\bibitem[\protect\citeauthoryear{Fischetti and Jo}{Fischetti and Jo}{2017}]%
        {fischetti2017deep}
\bibfield{author}{\bibinfo{person}{Matteo Fischetti} {and}
  \bibinfo{person}{Jason Jo}.} \bibinfo{year}{2017}\natexlab{}.
\newblock \showarticletitle{Deep Neural Networks as 0-1 Mixed Integer Linear
  Programs: A Feasibility Study}.
\newblock \bibinfo{journal}{\emph{arXiv preprint arXiv:1712.06174}}
  (\bibinfo{year}{2017}).
\newblock


\bibitem[\protect\citeauthoryear{Friedman, Popescu, et~al\mbox{.}}{Friedman
  et~al\mbox{.}}{2008}]%
        {friedman2008predictive}
\bibfield{author}{\bibinfo{person}{Jerome~H Friedman},
  \bibinfo{person}{Bogdan~E Popescu}, {et~al\mbox{.}}}
  \bibinfo{year}{2008}\natexlab{}.
\newblock \showarticletitle{Predictive learning via rule ensembles}.
\newblock \bibinfo{journal}{\emph{The Annals of Applied Statistics}}
  \bibinfo{volume}{2}, \bibinfo{number}{3} (\bibinfo{year}{2008}),
  \bibinfo{pages}{916--954}.
\newblock


\bibitem[\protect\citeauthoryear{Gehr, Mirman, Drachsler-Cohen, Tsankov,
  Chaudhuri, and Vechev}{Gehr et~al\mbox{.}}{2018}]%
        {gehrai}
\bibfield{author}{\bibinfo{person}{Timon Gehr}, \bibinfo{person}{Matthew
  Mirman}, \bibinfo{person}{Dana Drachsler-Cohen}, \bibinfo{person}{Petar
  Tsankov}, \bibinfo{person}{Swarat Chaudhuri}, {and} \bibinfo{person}{Martin
  Vechev}.} \bibinfo{year}{2018}\natexlab{}.
\newblock \showarticletitle{Ai 2: Safety and robustness certification of neural
  networks with abstract interpretation}. In \bibinfo{booktitle}{\emph{IEEE
  Symposium on Security and Privacy (SP)}}.
\newblock


\bibitem[\protect\citeauthoryear{Goncharov}{Goncharov}{[n.d.]}]%
        {tds}
\bibfield{author}{\bibinfo{person}{Maxim Goncharov}.}
  \bibinfo{year}{[n.d.]}\natexlab{}.
\newblock \bibinfo{title}{{Traffic direction systems as malware distribution
  tools}}.
\newblock
  \bibinfo{howpublished}{\url{http://www.trendmicro.es/media/misc/malware-distribution-tools-research-paper-en.pdf}}.
\newblock


\bibitem[\protect\citeauthoryear{Gouk, Frank, Pfahringer, and Cree}{Gouk
  et~al\mbox{.}}{2021}]%
        {gouk2021regularisation}
\bibfield{author}{\bibinfo{person}{Henry Gouk}, \bibinfo{person}{Eibe Frank},
  \bibinfo{person}{Bernhard Pfahringer}, {and} \bibinfo{person}{Michael~J
  Cree}.} \bibinfo{year}{2021}\natexlab{}.
\newblock \showarticletitle{Regularisation of neural networks by enforcing
  lipschitz continuity}.
\newblock \bibinfo{journal}{\emph{Machine Learning}} \bibinfo{volume}{110},
  \bibinfo{number}{2} (\bibinfo{year}{2021}), \bibinfo{pages}{393--416}.
\newblock


\bibitem[\protect\citeauthoryear{Grosse, Papernot, Manoharan, Backes, and
  McDaniel}{Grosse et~al\mbox{.}}{2016}]%
        {grosse2016adversarial}
\bibfield{author}{\bibinfo{person}{Kathrin Grosse}, \bibinfo{person}{Nicolas
  Papernot}, \bibinfo{person}{Praveen Manoharan}, \bibinfo{person}{Michael
  Backes}, {and} \bibinfo{person}{Patrick McDaniel}.}
  \bibinfo{year}{2016}\natexlab{}.
\newblock \showarticletitle{{Adversarial perturbations against deep neural
  networks for malware classification}}.
\newblock \bibinfo{journal}{\emph{arXiv preprint arXiv:1606.04435}}
  (\bibinfo{year}{2016}).
\newblock


\bibitem[\protect\citeauthoryear{Gupta, Cotter, Pfeifer, Voevodski, Canini,
  Mangylov, Moczydlowski, and Van~Esbroeck}{Gupta et~al\mbox{.}}{2016}]%
        {gupta2016monotonic}
\bibfield{author}{\bibinfo{person}{Maya Gupta}, \bibinfo{person}{Andrew
  Cotter}, \bibinfo{person}{Jan Pfeifer}, \bibinfo{person}{Konstantin
  Voevodski}, \bibinfo{person}{Kevin Canini}, \bibinfo{person}{Alexander
  Mangylov}, \bibinfo{person}{Wojciech Moczydlowski}, {and}
  \bibinfo{person}{Alexander Van~Esbroeck}.} \bibinfo{year}{2016}\natexlab{}.
\newblock \showarticletitle{Monotonic calibrated interpolated look-up tables}.
\newblock \bibinfo{journal}{\emph{The Journal of Machine Learning Research}}
  \bibinfo{volume}{17}, \bibinfo{number}{1} (\bibinfo{year}{2016}),
  \bibinfo{pages}{3790--3836}.
\newblock


\bibitem[\protect\citeauthoryear{Hein and Andriushchenko}{Hein and
  Andriushchenko}{2017}]%
        {hein2017formal}
\bibfield{author}{\bibinfo{person}{Matthias Hein} {and} \bibinfo{person}{Maksym
  Andriushchenko}.} \bibinfo{year}{2017}\natexlab{}.
\newblock \showarticletitle{Formal guarantees on the robustness of a classifier
  against adversarial manipulation}. In \bibinfo{booktitle}{\emph{Advances in
  Neural Information Processing Systems}}.
\newblock


\bibitem[\protect\citeauthoryear{Huang, Kwiatkowska, Wang, and Wu}{Huang
  et~al\mbox{.}}{2017}]%
        {huang2017safety}
\bibfield{author}{\bibinfo{person}{Xiaowei Huang}, \bibinfo{person}{Marta
  Kwiatkowska}, \bibinfo{person}{Sen Wang}, {and} \bibinfo{person}{Min Wu}.}
  \bibinfo{year}{2017}\natexlab{}.
\newblock \showarticletitle{Safety verification of deep neural networks}. In
  \bibinfo{booktitle}{\emph{International Conference on Computer Aided
  Verification (CAV)}}. Springer, \bibinfo{pages}{3--29}.
\newblock


\bibitem[\protect\citeauthoryear{Incer, Theodorides, Afroz, and Wagner}{Incer
  et~al\mbox{.}}{2018}]%
        {incer2018adversarially}
\bibfield{author}{\bibinfo{person}{Inigo Incer}, \bibinfo{person}{Michael
  Theodorides}, \bibinfo{person}{Sadia Afroz}, {and} \bibinfo{person}{David
  Wagner}.} \bibinfo{year}{2018}\natexlab{}.
\newblock \showarticletitle{Adversarially Robust Malware Detection Using
  Monotonic Classification}. In \bibinfo{booktitle}{\emph{Proceedings of the
  Fourth ACM International Workshop on Security and Privacy Analytics}}. ACM,
  \bibinfo{pages}{54--63}.
\newblock


\bibitem[\protect\citeauthoryear{Jia, Cao, Wang, and Gong}{Jia
  et~al\mbox{.}}{2019}]%
        {jia2019certified}
\bibfield{author}{\bibinfo{person}{Jinyuan Jia}, \bibinfo{person}{Xiaoyu Cao},
  \bibinfo{person}{Binghui Wang}, {and} \bibinfo{person}{Neil~Zhenqiang Gong}.}
  \bibinfo{year}{2019}\natexlab{}.
\newblock \showarticletitle{Certified robustness for top-k predictions against
  adversarial perturbations via randomized smoothing}.
\newblock \bibinfo{journal}{\emph{International Conference on Learning
  Representations (ICLR)}} (\bibinfo{year}{2019}).
\newblock


\bibitem[\protect\citeauthoryear{Kantchelian, Tygar, and Joseph}{Kantchelian
  et~al\mbox{.}}{2016}]%
        {kantchelian2016evasion}
\bibfield{author}{\bibinfo{person}{Alex Kantchelian}, \bibinfo{person}{JD
  Tygar}, {and} \bibinfo{person}{Anthony Joseph}.}
  \bibinfo{year}{2016}\natexlab{}.
\newblock \showarticletitle{Evasion and hardening of tree ensemble
  classifiers}. In \bibinfo{booktitle}{\emph{International Conference on
  Machine Learning}}. \bibinfo{pages}{2387--2396}.
\newblock


\bibitem[\protect\citeauthoryear{Katz, Barrett, Dill, Julian, and
  Kochenderfer}{Katz et~al\mbox{.}}{2017}]%
        {katz2017reluplex}
\bibfield{author}{\bibinfo{person}{Guy Katz}, \bibinfo{person}{Clark Barrett},
  \bibinfo{person}{David~L Dill}, \bibinfo{person}{Kyle Julian}, {and}
  \bibinfo{person}{Mykel~J Kochenderfer}.} \bibinfo{year}{2017}\natexlab{}.
\newblock \showarticletitle{Reluplex: An efficient SMT solver for verifying
  deep neural networks}. In \bibinfo{booktitle}{\emph{International Conference
  on Computer Aided Verification (CAV)}}. Springer, \bibinfo{pages}{97--117}.
\newblock


\bibitem[\protect\citeauthoryear{Katz, Huang, Ibeling, Julian, Lazarus, Lim,
  Shah, Thakoor, Wu, Zelji{\'c}, et~al\mbox{.}}{Katz et~al\mbox{.}}{2019}]%
        {katz2019marabou}
\bibfield{author}{\bibinfo{person}{Guy Katz}, \bibinfo{person}{Derek~A Huang},
  \bibinfo{person}{Duligur Ibeling}, \bibinfo{person}{Kyle Julian},
  \bibinfo{person}{Christopher Lazarus}, \bibinfo{person}{Rachel Lim},
  \bibinfo{person}{Parth Shah}, \bibinfo{person}{Shantanu Thakoor},
  \bibinfo{person}{Haoze Wu}, \bibinfo{person}{Aleksandar Zelji{\'c}},
  {et~al\mbox{.}}} \bibinfo{year}{2019}\natexlab{}.
\newblock \showarticletitle{The marabou framework for verification and analysis
  of deep neural networks}. In \bibinfo{booktitle}{\emph{International
  Conference on Computer Aided Verification (CAV)}}.
\newblock


\bibitem[\protect\citeauthoryear{Kay and Ungar}{Kay and Ungar}{2000}]%
        {kay2000estimating}
\bibfield{author}{\bibinfo{person}{Herbert Kay} {and} \bibinfo{person}{Lyle~H
  Ungar}.} \bibinfo{year}{2000}\natexlab{}.
\newblock \showarticletitle{Estimating monotonic functions and their bounds}.
\newblock \bibinfo{journal}{\emph{AIChE Journal}} \bibinfo{volume}{46},
  \bibinfo{number}{12} (\bibinfo{year}{2000}), \bibinfo{pages}{2426--2434}.
\newblock


\bibitem[\protect\citeauthoryear{Kharraz, Ma, Murley, Lever, Mason, Miller,
  Borisov, Antonakakis, and Bailey}{Kharraz et~al\mbox{.}}{2019}]%
        {kharraz2019outguard}
\bibfield{author}{\bibinfo{person}{Amin Kharraz}, \bibinfo{person}{Zane Ma},
  \bibinfo{person}{Paul Murley}, \bibinfo{person}{Charles Lever},
  \bibinfo{person}{Joshua Mason}, \bibinfo{person}{Andrew Miller},
  \bibinfo{person}{Nikita Borisov}, \bibinfo{person}{Manos Antonakakis}, {and}
  \bibinfo{person}{Michael Bailey}.} \bibinfo{year}{2019}\natexlab{}.
\newblock \showarticletitle{Outguard: Detecting in-browser covert
  cryptocurrency mining in the wild}. In \bibinfo{booktitle}{\emph{The World
  Wide Web Conference}}. \bibinfo{pages}{840--852}.
\newblock


\bibitem[\protect\citeauthoryear{Kurakin, Goodfellow, and Bengio}{Kurakin
  et~al\mbox{.}}{2017}]%
        {kurakin2016adversarial}
\bibfield{author}{\bibinfo{person}{Alexey Kurakin}, \bibinfo{person}{Ian
  Goodfellow}, {and} \bibinfo{person}{Samy Bengio}.}
  \bibinfo{year}{2017}\natexlab{}.
\newblock \showarticletitle{Adversarial machine learning at scale}. In
  \bibinfo{booktitle}{\emph{International Conference on Learning
  Representations (ICLR)}}.
\newblock


\bibitem[\protect\citeauthoryear{Kwon, Baig, and Akoglu}{Kwon
  et~al\mbox{.}}{2017}]%
        {kwon2017domain}
\bibfield{author}{\bibinfo{person}{Heeyoung Kwon}, \bibinfo{person}{Mirza~Basim
  Baig}, {and} \bibinfo{person}{Leman Akoglu}.}
  \bibinfo{year}{2017}\natexlab{}.
\newblock \showarticletitle{A domain-agnostic approach to spam-URL detection
  via redirects}. In \bibinfo{booktitle}{\emph{Pacific-Asia Conference on
  Knowledge Discovery and Data Mining}}. Springer, \bibinfo{pages}{220--232}.
\newblock


\bibitem[\protect\citeauthoryear{Laskov et~al\mbox{.}}{Laskov
  et~al\mbox{.}}{2014}]%
        {laskov2014practical}
\bibfield{author}{\bibinfo{person}{Pavel Laskov} {et~al\mbox{.}}}
  \bibinfo{year}{2014}\natexlab{}.
\newblock \showarticletitle{Practical evasion of a learning-based classifier: A
  case study}. In \bibinfo{booktitle}{\emph{Security and Privacy (SP), 2014
  IEEE Symposium on}}. IEEE, \bibinfo{pages}{197--211}.
\newblock


\bibitem[\protect\citeauthoryear{Lecuyer, Atlidakis, Geambasu, Hsu, and
  Jana}{Lecuyer et~al\mbox{.}}{2019}]%
        {lecuyer2019certified}
\bibfield{author}{\bibinfo{person}{Mathias Lecuyer}, \bibinfo{person}{Vaggelis
  Atlidakis}, \bibinfo{person}{Roxana Geambasu}, \bibinfo{person}{Daniel Hsu},
  {and} \bibinfo{person}{Suman Jana}.} \bibinfo{year}{2019}\natexlab{}.
\newblock \showarticletitle{Certified robustness to adversarial examples with
  differential privacy}. In \bibinfo{booktitle}{\emph{2019 IEEE Symposium on
  Security and Privacy (SP)}}. IEEE.
\newblock


\bibitem[\protect\citeauthoryear{Lee, Caverlee, and Webb}{Lee
  et~al\mbox{.}}{2010}]%
        {lee2010uncovering}
\bibfield{author}{\bibinfo{person}{Kyumin Lee}, \bibinfo{person}{James
  Caverlee}, {and} \bibinfo{person}{Steve Webb}.}
  \bibinfo{year}{2010}\natexlab{}.
\newblock \showarticletitle{Uncovering social spammers: social honeypots+
  machine learning}. In \bibinfo{booktitle}{\emph{Proceedings of the 33rd
  international ACM SIGIR conference on Research and development in information
  retrieval}}. \bibinfo{pages}{435--442}.
\newblock


\bibitem[\protect\citeauthoryear{Lee, Eoff, and Caverlee}{Lee
  et~al\mbox{.}}{2011}]%
        {lee2011seven}
\bibfield{author}{\bibinfo{person}{Kyumin Lee}, \bibinfo{person}{Brian Eoff},
  {and} \bibinfo{person}{James Caverlee}.} \bibinfo{year}{2011}\natexlab{}.
\newblock \showarticletitle{Seven months with the devils: A long-term study of
  content polluters on twitter}. In \bibinfo{booktitle}{\emph{Proceedings of
  the International AAAI Conference on Web and Social Media}},
  Vol.~\bibinfo{volume}{5}.
\newblock


\bibitem[\protect\citeauthoryear{Lee and Kim}{Lee and Kim}{2013}]%
        {lee2013warningbird}
\bibfield{author}{\bibinfo{person}{Sangho Lee} {and} \bibinfo{person}{Jong
  Kim}.} \bibinfo{year}{2013}\natexlab{}.
\newblock \showarticletitle{Warningbird: A near real-time detection system for
  suspicious urls in twitter stream}.
\newblock \bibinfo{journal}{\emph{IEEE transactions on dependable and secure
  computing}} \bibinfo{volume}{10}, \bibinfo{number}{3} (\bibinfo{year}{2013}),
  \bibinfo{pages}{183--195}.
\newblock


\bibitem[\protect\citeauthoryear{Lee, Lee, and Park}{Lee et~al\mbox{.}}{2020}]%
        {lee2020lipschitz}
\bibfield{author}{\bibinfo{person}{Sungyoon Lee}, \bibinfo{person}{Jaewook
  Lee}, {and} \bibinfo{person}{Saerom Park}.} \bibinfo{year}{2020}\natexlab{}.
\newblock \showarticletitle{Lipschitz-Certifiable Training with a Tight Outer
  Bound}.
\newblock \bibinfo{journal}{\emph{Advances in Neural Information Processing
  Systems (NeurIPS)}} (\bibinfo{year}{2020}).
\newblock


\bibitem[\protect\citeauthoryear{Leino, Wang, and Fredrikson}{Leino
  et~al\mbox{.}}{2021}]%
        {leino2021globally}
\bibfield{author}{\bibinfo{person}{Klas Leino}, \bibinfo{person}{Zifan Wang},
  {and} \bibinfo{person}{Matt Fredrikson}.} \bibinfo{year}{2021}\natexlab{}.
\newblock \showarticletitle{Globally-Robust Neural Networks}.
\newblock \bibinfo{journal}{\emph{arXiv preprint arXiv:2102.08452}}
  (\bibinfo{year}{2021}).
\newblock


\bibitem[\protect\citeauthoryear{Li, Chen, Wang, and Carin}{Li
  et~al\mbox{.}}{2018}]%
        {li2018second}
\bibfield{author}{\bibinfo{person}{Bai Li}, \bibinfo{person}{Changyou Chen},
  \bibinfo{person}{Wenlin Wang}, {and} \bibinfo{person}{Lawrence Carin}.}
  \bibinfo{year}{2018}\natexlab{}.
\newblock \showarticletitle{Second-order adversarial attack and certifiable
  robustness}.
\newblock  (\bibinfo{year}{2018}).
\newblock


\bibitem[\protect\citeauthoryear{Li, Chen, Wang, and Carin}{Li
  et~al\mbox{.}}{2019a}]%
        {li2018certified}
\bibfield{author}{\bibinfo{person}{Bai Li}, \bibinfo{person}{Changyou Chen},
  \bibinfo{person}{Wenlin Wang}, {and} \bibinfo{person}{Lawrence Carin}.}
  \bibinfo{year}{2019}\natexlab{a}.
\newblock \showarticletitle{Certified adversarial robustness with additive
  noise}.
\newblock \bibinfo{journal}{\emph{Advances in Neural Information Processing
  Systems (NeurIPS)}} (\bibinfo{year}{2019}).
\newblock


\bibitem[\protect\citeauthoryear{Li, Qi, Xie, and Li}{Li et~al\mbox{.}}{2020}]%
        {li2020sok}
\bibfield{author}{\bibinfo{person}{Linyi Li}, \bibinfo{person}{Xiangyu Qi},
  \bibinfo{person}{Tao Xie}, {and} \bibinfo{person}{Bo Li}.}
  \bibinfo{year}{2020}\natexlab{}.
\newblock \showarticletitle{SoK: Certified Robustness for Deep Neural
  Networks}.
\newblock \bibinfo{journal}{\emph{arXiv preprint arXiv:2009.04131}}
  (\bibinfo{year}{2020}).
\newblock


\bibitem[\protect\citeauthoryear{Li, Zhong, Li, and Xie}{Li
  et~al\mbox{.}}{2019b}]%
        {li2019robustra}
\bibfield{author}{\bibinfo{person}{Linyi Li}, \bibinfo{person}{Zexuan Zhong},
  \bibinfo{person}{Bo Li}, {and} \bibinfo{person}{Tao Xie}.}
  \bibinfo{year}{2019}\natexlab{b}.
\newblock \showarticletitle{Robustra: Training Provable Robust Neural Networks
  over Reference Adversarial Space}. In \bibinfo{booktitle}{\emph{IJCAI}}.
\newblock


\bibitem[\protect\citeauthoryear{Li, Alrwais, Xie, Yu, and Wang}{Li
  et~al\mbox{.}}{2013}]%
        {li2013finding}
\bibfield{author}{\bibinfo{person}{Zhou Li}, \bibinfo{person}{Sumayah Alrwais},
  \bibinfo{person}{Yinglian Xie}, \bibinfo{person}{Fang Yu}, {and}
  \bibinfo{person}{XiaoFeng Wang}.} \bibinfo{year}{2013}\natexlab{}.
\newblock \showarticletitle{Finding the linchpins of the dark web: a study on
  topologically dedicated hosts on malicious web infrastructures}. In
  \bibinfo{booktitle}{\emph{2013 IEEE Symposium on Security and Privacy}}.
  IEEE, \bibinfo{pages}{112--126}.
\newblock


\bibitem[\protect\citeauthoryear{Lin, Zhu, Samanta, and Jagannathan}{Lin
  et~al\mbox{.}}{2020}]%
        {lin2020art}
\bibfield{author}{\bibinfo{person}{Xuankang Lin}, \bibinfo{person}{He Zhu},
  \bibinfo{person}{Roopsha Samanta}, {and} \bibinfo{person}{Suresh
  Jagannathan}.} \bibinfo{year}{2020}\natexlab{}.
\newblock \showarticletitle{ART: abstraction refinement-guided training for
  provably correct neural networks}. In \bibinfo{booktitle}{\emph{2020 Formal
  Methods in Computer Aided Design (FMCAD)}}. IEEE, \bibinfo{pages}{148--157}.
\newblock


\bibitem[\protect\citeauthoryear{Lomuscio and Maganti}{Lomuscio and
  Maganti}{2017}]%
        {lomuscio2017approach}
\bibfield{author}{\bibinfo{person}{Alessio Lomuscio} {and}
  \bibinfo{person}{Lalit Maganti}.} \bibinfo{year}{2017}\natexlab{}.
\newblock \showarticletitle{An approach to reachability analysis for
  feed-forward relu neural networks}.
\newblock \bibinfo{journal}{\emph{arXiv preprint arXiv:1706.07351}}
  (\bibinfo{year}{2017}).
\newblock


\bibitem[\protect\citeauthoryear{Ma, Saul, Savage, and Voelker}{Ma
  et~al\mbox{.}}{2009}]%
        {ma2009identifying}
\bibfield{author}{\bibinfo{person}{Justin Ma}, \bibinfo{person}{Lawrence~K
  Saul}, \bibinfo{person}{Stefan Savage}, {and} \bibinfo{person}{Geoffrey~M
  Voelker}.} \bibinfo{year}{2009}\natexlab{}.
\newblock \showarticletitle{Identifying suspicious URLs: an application of
  large-scale online learning}. In \bibinfo{booktitle}{\emph{Proceedings of the
  26th annual international conference on machine learning}}.
  \bibinfo{pages}{681--688}.
\newblock


\bibitem[\protect\citeauthoryear{Madry, Makelov, Schmidt, Tsipras, and
  Vladu}{Madry et~al\mbox{.}}{2018}]%
        {madry2017towards}
\bibfield{author}{\bibinfo{person}{Aleksander Madry},
  \bibinfo{person}{Aleksandar Makelov}, \bibinfo{person}{Ludwig Schmidt},
  \bibinfo{person}{Dimitris Tsipras}, {and} \bibinfo{person}{Adrian Vladu}.}
  \bibinfo{year}{2018}\natexlab{}.
\newblock \showarticletitle{Towards deep learning models resistant to
  adversarial attacks}.
\newblock \bibinfo{journal}{\emph{International Conference on Learning
  Representations (ICLR)}} (\bibinfo{year}{2018}).
\newblock


\bibitem[\protect\citeauthoryear{Melacci, Ciravegna, Sotgiu, Demontis, Biggio,
  Gori, and Roli}{Melacci et~al\mbox{.}}{2020}]%
        {melacci2020can}
\bibfield{author}{\bibinfo{person}{Stefano Melacci}, \bibinfo{person}{Gabriele
  Ciravegna}, \bibinfo{person}{Angelo Sotgiu}, \bibinfo{person}{Ambra
  Demontis}, \bibinfo{person}{Battista Biggio}, \bibinfo{person}{Marco Gori},
  {and} \bibinfo{person}{Fabio Roli}.} \bibinfo{year}{2020}\natexlab{}.
\newblock \showarticletitle{Can Domain Knowledge Alleviate Adversarial Attacks
  in Multi-Label Classifiers?}
\newblock \bibinfo{journal}{\emph{arXiv preprint arXiv:2006.03833}}
  (\bibinfo{year}{2020}).
\newblock


\bibitem[\protect\citeauthoryear{Miller, Kantchelian, Tschantz, Afroz,
  Bachwani, Faizullabhoy, Huang, Shankar, Wu, Yiu, et~al\mbox{.}}{Miller
  et~al\mbox{.}}{2016}]%
        {miller2016reviewer}
\bibfield{author}{\bibinfo{person}{Brad Miller}, \bibinfo{person}{Alex
  Kantchelian}, \bibinfo{person}{Michael~Carl Tschantz}, \bibinfo{person}{Sadia
  Afroz}, \bibinfo{person}{Rekha Bachwani}, \bibinfo{person}{Riyaz
  Faizullabhoy}, \bibinfo{person}{Ling Huang}, \bibinfo{person}{Vaishaal
  Shankar}, \bibinfo{person}{Tony Wu}, \bibinfo{person}{George Yiu},
  {et~al\mbox{.}}} \bibinfo{year}{2016}\natexlab{}.
\newblock \showarticletitle{Reviewer integration and performance measurement
  for malware detection}. In \bibinfo{booktitle}{\emph{International Conference
  on Detection of Intrusions and Malware, and Vulnerability Assessment}}.
  Springer, \bibinfo{pages}{122--141}.
\newblock


\bibitem[\protect\citeauthoryear{Mirman, Gehr, and Vechev}{Mirman
  et~al\mbox{.}}{2018}]%
        {mirman2018differentiable}
\bibfield{author}{\bibinfo{person}{Matthew Mirman}, \bibinfo{person}{Timon
  Gehr}, {and} \bibinfo{person}{Martin Vechev}.}
  \bibinfo{year}{2018}\natexlab{}.
\newblock \showarticletitle{Differentiable Abstract Interpretation for Provably
  Robust Neural Networks}. In \bibinfo{booktitle}{\emph{International
  Conference on Machine Learning (ICML)}}. \bibinfo{pages}{3575--3583}.
\newblock


\bibitem[\protect\citeauthoryear{M{\"u}ller, Serre, Singh, P{\"u}schel, and
  Vechev}{M{\"u}ller et~al\mbox{.}}{2021b}]%
        {muller2021scaling}
\bibfield{author}{\bibinfo{person}{Christoph M{\"u}ller},
  \bibinfo{person}{Fran{\c{c}}ois Serre}, \bibinfo{person}{Gagandeep Singh},
  \bibinfo{person}{Markus P{\"u}schel}, {and} \bibinfo{person}{Martin Vechev}.}
  \bibinfo{year}{2021}\natexlab{b}.
\newblock \showarticletitle{Scaling Polyhedral Neural Network Verification on
  GPUs}.
\newblock \bibinfo{journal}{\emph{Proceedings of Machine Learning and Systems}}
   \bibinfo{volume}{3} (\bibinfo{year}{2021}).
\newblock


\bibitem[\protect\citeauthoryear{M{\"u}ller, Makarchuk, Singh, P{\"u}schel, and
  Vechev}{M{\"u}ller et~al\mbox{.}}{2021a}]%
        {muller2021precise}
\bibfield{author}{\bibinfo{person}{Mark~Niklas M{\"u}ller},
  \bibinfo{person}{Gleb Makarchuk}, \bibinfo{person}{Gagandeep Singh},
  \bibinfo{person}{Markus P{\"u}schel}, {and} \bibinfo{person}{Martin Vechev}.}
  \bibinfo{year}{2021}\natexlab{a}.
\newblock \showarticletitle{Precise Multi-Neuron Abstractions for Neural
  Network Certification}.
\newblock \bibinfo{journal}{\emph{arXiv preprint arXiv:2103.03638}}
  (\bibinfo{year}{2021}).
\newblock


\bibitem[\protect\citeauthoryear{Pauli, Koch, Berberich, Kohler, and
  Allgower}{Pauli et~al\mbox{.}}{2021}]%
        {pauli2021training}
\bibfield{author}{\bibinfo{person}{Patricia Pauli}, \bibinfo{person}{Anne
  Koch}, \bibinfo{person}{Julian Berberich}, \bibinfo{person}{Paul Kohler},
  {and} \bibinfo{person}{Frank Allgower}.} \bibinfo{year}{2021}\natexlab{}.
\newblock \showarticletitle{Training robust neural networks using Lipschitz
  bounds}.
\newblock \bibinfo{journal}{\emph{IEEE Control Systems Letters}}
  (\bibinfo{year}{2021}).
\newblock


\bibitem[\protect\citeauthoryear{Pendlebury, Pierazzi, Jordaney, Kinder, and
  Cavallaro}{Pendlebury et~al\mbox{.}}{2019}]%
        {pendlebury2019tesseract}
\bibfield{author}{\bibinfo{person}{Feargus Pendlebury}, \bibinfo{person}{Fabio
  Pierazzi}, \bibinfo{person}{Roberto Jordaney}, \bibinfo{person}{Johannes
  Kinder}, {and} \bibinfo{person}{Lorenzo Cavallaro}.}
  \bibinfo{year}{2019}\natexlab{}.
\newblock \showarticletitle{{TESSERACT: Eliminating experimental bias in
  malware classification across space and time}}. In
  \bibinfo{booktitle}{\emph{28th USENIX Security Symposium (USENIX Security
  19)}}. \bibinfo{pages}{729--746}.
\newblock


\bibitem[\protect\citeauthoryear{Pierazzi, Pendlebury, Cortellazzi, and
  Cavallaro}{Pierazzi et~al\mbox{.}}{2020}]%
        {pierazzi2020problemspace}
\bibfield{author}{\bibinfo{person}{F. Pierazzi}, \bibinfo{person}{F.
  Pendlebury}, \bibinfo{person}{J. Cortellazzi}, {and} \bibinfo{person}{L.
  Cavallaro}.} \bibinfo{year}{2020}\natexlab{}.
\newblock \showarticletitle{Intriguing Properties of Adversarial ML Attacks in
  the Problem Space}. In \bibinfo{booktitle}{\emph{2020 IEEE Symposium on
  Security and Privacy (SP)}}. \bibinfo{publisher}{IEEE Computer Society},
  \bibinfo{pages}{1308--1325}.
\newblock
\showISSN{2375-1207}
\urldef\tempurl%
\url{https://doi.org/10.1109/SP40000.2020.00073}
\showDOI{\tempurl}


\bibitem[\protect\citeauthoryear{Qin, Martens, Gowal, Krishnan, Dvijotham,
  Fawzi, De, Stanforth, and Kohli}{Qin et~al\mbox{.}}{2019}]%
        {qin2019adversarial}
\bibfield{author}{\bibinfo{person}{Chongli Qin}, \bibinfo{person}{James
  Martens}, \bibinfo{person}{Sven Gowal}, \bibinfo{person}{Dilip Krishnan},
  \bibinfo{person}{Krishnamurthy Dvijotham}, \bibinfo{person}{Alhussein Fawzi},
  \bibinfo{person}{Soham De}, \bibinfo{person}{Robert Stanforth}, {and}
  \bibinfo{person}{Pushmeet Kohli}.} \bibinfo{year}{2019}\natexlab{}.
\newblock \showarticletitle{Adversarial robustness through local
  linearization}.
\newblock \bibinfo{journal}{\emph{Advances in Neural Information Processing
  Systems (NIPS)}} (\bibinfo{year}{2019}).
\newblock


\bibitem[\protect\citeauthoryear{Raghothaman, Mendelson, Zhao, Naik, and
  Scholz}{Raghothaman et~al\mbox{.}}{2019}]%
        {raghothaman2019provenance}
\bibfield{author}{\bibinfo{person}{Mukund Raghothaman},
  \bibinfo{person}{Jonathan Mendelson}, \bibinfo{person}{David Zhao},
  \bibinfo{person}{Mayur Naik}, {and} \bibinfo{person}{Bernhard Scholz}.}
  \bibinfo{year}{2019}\natexlab{}.
\newblock \showarticletitle{Provenance-guided synthesis of Datalog programs}.
\newblock \bibinfo{journal}{\emph{Proceedings of the ACM on Programming
  Languages}} \bibinfo{volume}{4}, \bibinfo{number}{POPL}
  (\bibinfo{year}{2019}), \bibinfo{pages}{1--27}.
\newblock


\bibitem[\protect\citeauthoryear{Raghunathan, Steinhardt, and
  Liang}{Raghunathan et~al\mbox{.}}{2018a}]%
        {raghunathan2018certified}
\bibfield{author}{\bibinfo{person}{Aditi Raghunathan}, \bibinfo{person}{Jacob
  Steinhardt}, {and} \bibinfo{person}{Percy Liang}.}
  \bibinfo{year}{2018}\natexlab{a}.
\newblock \showarticletitle{Certified defenses against adversarial examples}.
\newblock \bibinfo{journal}{\emph{International Conference on Learning
  Representations (ICLR)}} (\bibinfo{year}{2018}).
\newblock


\bibitem[\protect\citeauthoryear{Raghunathan, Steinhardt, and
  Liang}{Raghunathan et~al\mbox{.}}{2018b}]%
        {raghunathan2018semidefinite}
\bibfield{author}{\bibinfo{person}{Aditi Raghunathan}, \bibinfo{person}{Jacob
  Steinhardt}, {and} \bibinfo{person}{Percy~S Liang}.}
  \bibinfo{year}{2018}\natexlab{b}.
\newblock \showarticletitle{Semidefinite relaxations for certifying robustness
  to adversarial examples}. In \bibinfo{booktitle}{\emph{Advances in Neural
  Information Processing Systems}}. \bibinfo{pages}{10900--10910}.
\newblock


\bibitem[\protect\citeauthoryear{Ryan, Wong, Yao, Gu, and Jana}{Ryan
  et~al\mbox{.}}{2020}]%
        {ryan2019cln2inv}
\bibfield{author}{\bibinfo{person}{Gabriel Ryan}, \bibinfo{person}{Justin
  Wong}, \bibinfo{person}{Jianan Yao}, \bibinfo{person}{Ronghui Gu}, {and}
  \bibinfo{person}{Suman Jana}.} \bibinfo{year}{2020}\natexlab{}.
\newblock \showarticletitle{CLN2INV: Learning Loop Invariants with Continuous
  Logic Networks}. In \bibinfo{booktitle}{\emph{International Conference on
  Learning Representations (ICLR)}}.
\newblock


\bibitem[\protect\citeauthoryear{Salman, Yang, Li, Zhang, Zhang, Razenshteyn,
  and Bubeck}{Salman et~al\mbox{.}}{2019a}]%
        {salman2019provably}
\bibfield{author}{\bibinfo{person}{Hadi Salman}, \bibinfo{person}{Greg Yang},
  \bibinfo{person}{Jerry Li}, \bibinfo{person}{Pengchuan Zhang},
  \bibinfo{person}{Huan Zhang}, \bibinfo{person}{Ilya Razenshteyn}, {and}
  \bibinfo{person}{Sebastien Bubeck}.} \bibinfo{year}{2019}\natexlab{a}.
\newblock \showarticletitle{Provably robust deep learning via adversarially
  trained smoothed classifiers}.
\newblock \bibinfo{journal}{\emph{Advances in Neural Information Processing
  Systems (NeurIPS)}} (\bibinfo{year}{2019}).
\newblock


\bibitem[\protect\citeauthoryear{Salman, Yang, Zhang, Hsieh, and Zhang}{Salman
  et~al\mbox{.}}{2019b}]%
        {salman2019convex}
\bibfield{author}{\bibinfo{person}{Hadi Salman}, \bibinfo{person}{Greg Yang},
  \bibinfo{person}{Huan Zhang}, \bibinfo{person}{Cho-Jui Hsieh}, {and}
  \bibinfo{person}{Pengchuan Zhang}.} \bibinfo{year}{2019}\natexlab{b}.
\newblock \showarticletitle{A convex relaxation barrier to tight robustness
  verification of neural networks}.
\newblock \bibinfo{journal}{\emph{Advances in Neural Information Processing
  Systems (NeurIPS)}} (\bibinfo{year}{2019}).
\newblock


\bibitem[\protect\citeauthoryear{Si, Lee, Zhang, Albarghouthi, Koutris, and
  Naik}{Si et~al\mbox{.}}{2018}]%
        {si2018syntax}
\bibfield{author}{\bibinfo{person}{Xujie Si}, \bibinfo{person}{Woosuk Lee},
  \bibinfo{person}{Richard Zhang}, \bibinfo{person}{Aws Albarghouthi},
  \bibinfo{person}{Paraschos Koutris}, {and} \bibinfo{person}{Mayur Naik}.}
  \bibinfo{year}{2018}\natexlab{}.
\newblock \showarticletitle{Syntax-guided synthesis of datalog programs}. In
  \bibinfo{booktitle}{\emph{Proceedings of the 2018 26th ACM Joint Meeting on
  European Software Engineering Conference and Symposium on the Foundations of
  Software Engineering}}. \bibinfo{pages}{515--527}.
\newblock


\bibitem[\protect\citeauthoryear{Singh, Ganvir, P{\"u}schel, and Vechev}{Singh
  et~al\mbox{.}}{2019a}]%
        {singh2019beyond}
\bibfield{author}{\bibinfo{person}{Gagandeep Singh}, \bibinfo{person}{Rupanshu
  Ganvir}, \bibinfo{person}{Markus P{\"u}schel}, {and} \bibinfo{person}{Martin
  Vechev}.} \bibinfo{year}{2019}\natexlab{a}.
\newblock \showarticletitle{Beyond the single neuron convex barrier for neural
  network certification}.
\newblock \bibinfo{journal}{\emph{Advances in Neural Information Processing
  Systems (NeurIPS)}} (\bibinfo{year}{2019}).
\newblock


\bibitem[\protect\citeauthoryear{Singh, Gehr, P{\"u}schel, and Vechev}{Singh
  et~al\mbox{.}}{2019b}]%
        {singh2019abstract}
\bibfield{author}{\bibinfo{person}{Gagandeep Singh}, \bibinfo{person}{Timon
  Gehr}, \bibinfo{person}{Markus P{\"u}schel}, {and} \bibinfo{person}{Martin
  Vechev}.} \bibinfo{year}{2019}\natexlab{b}.
\newblock \showarticletitle{An abstract domain for certifying neural networks}.
\newblock \bibinfo{journal}{\emph{Proceedings of the ACM on Programming
  Languages}} \bibinfo{volume}{3}, \bibinfo{number}{POPL}
  (\bibinfo{year}{2019}), \bibinfo{pages}{1--30}.
\newblock


\bibitem[\protect\citeauthoryear{Singh, Gehr, P{\"u}schel, and Vechev}{Singh
  et~al\mbox{.}}{2019c}]%
        {singh2019boosting}
\bibfield{author}{\bibinfo{person}{Gagandeep Singh}, \bibinfo{person}{Timon
  Gehr}, \bibinfo{person}{Markus P{\"u}schel}, {and} \bibinfo{person}{Martin~T
  Vechev}.} \bibinfo{year}{2019}\natexlab{c}.
\newblock \showarticletitle{Boosting Robustness Certification of Neural
  Networks.}. In \bibinfo{booktitle}{\emph{ICLR (Poster)}}.
\newblock


\bibitem[\protect\citeauthoryear{Singla and Feizi}{Singla and Feizi}{2019}]%
        {singla2019bounding}
\bibfield{author}{\bibinfo{person}{Sahil Singla} {and} \bibinfo{person}{Soheil
  Feizi}.} \bibinfo{year}{2019}\natexlab{}.
\newblock \showarticletitle{Bounding singular values of convolution layers}.
\newblock \bibinfo{journal}{\emph{arXiv preprint arXiv:1911.10258}}
  (\bibinfo{year}{2019}).
\newblock


\bibitem[\protect\citeauthoryear{Solar-Lezama, Tancau, Bodik, Seshia, and
  Saraswat}{Solar-Lezama et~al\mbox{.}}{2006}]%
        {solar2006combinatorial}
\bibfield{author}{\bibinfo{person}{Armando Solar-Lezama},
  \bibinfo{person}{Liviu Tancau}, \bibinfo{person}{Rastislav Bodik},
  \bibinfo{person}{Sanjit Seshia}, {and} \bibinfo{person}{Vijay Saraswat}.}
  \bibinfo{year}{2006}\natexlab{}.
\newblock \showarticletitle{Combinatorial sketching for finite programs}. In
  \bibinfo{booktitle}{\emph{Proceedings of the 12th international conference on
  Architectural support for programming languages and operating systems}}.
  \bibinfo{pages}{404--415}.
\newblock


\bibitem[\protect\citeauthoryear{Szegedy, Zaremba, Sutskever, Bruna, Erhan,
  Goodfellow, and Fergus}{Szegedy et~al\mbox{.}}{2013}]%
        {szegedy2013intriguing}
\bibfield{author}{\bibinfo{person}{Christian Szegedy},
  \bibinfo{person}{Wojciech Zaremba}, \bibinfo{person}{Ilya Sutskever},
  \bibinfo{person}{Joan Bruna}, \bibinfo{person}{Dumitru Erhan},
  \bibinfo{person}{Ian Goodfellow}, {and} \bibinfo{person}{Rob Fergus}.}
  \bibinfo{year}{2013}\natexlab{}.
\newblock \showarticletitle{Intriguing properties of neural networks}.
\newblock \bibinfo{journal}{\emph{International Conference on Learning
  Representations (ICLR)}} (\bibinfo{year}{2013}).
\newblock


\bibitem[\protect\citeauthoryear{Thomas, Grier, Ma, Paxson, and Song}{Thomas
  et~al\mbox{.}}{2011}]%
        {thomas2011design}
\bibfield{author}{\bibinfo{person}{Kurt Thomas}, \bibinfo{person}{Chris Grier},
  \bibinfo{person}{Justin Ma}, \bibinfo{person}{Vern Paxson}, {and}
  \bibinfo{person}{Dawn Song}.} \bibinfo{year}{2011}\natexlab{}.
\newblock \showarticletitle{Design and evaluation of a real-time url spam
  filtering service}. In \bibinfo{booktitle}{\emph{2011 IEEE symposium on
  security and privacy}}. IEEE, \bibinfo{pages}{447--462}.
\newblock


\bibitem[\protect\citeauthoryear{Tjeng, Xiao, and Tedrake}{Tjeng
  et~al\mbox{.}}{2017}]%
        {tjeng2017evaluating}
\bibfield{author}{\bibinfo{person}{Vincent Tjeng}, \bibinfo{person}{Kai Xiao},
  {and} \bibinfo{person}{Russ Tedrake}.} \bibinfo{year}{2017}\natexlab{}.
\newblock \showarticletitle{Evaluating Robustness of Neural Networks with Mixed
  Integer Programming}.
\newblock \bibinfo{journal}{\emph{arXiv preprint arXiv:1711.07356}}
  (\bibinfo{year}{2017}).
\newblock


\bibitem[\protect\citeauthoryear{Wagner and Soto}{Wagner and Soto}{2002}]%
        {wagner2002mimicry}
\bibfield{author}{\bibinfo{person}{David Wagner} {and} \bibinfo{person}{Paolo
  Soto}.} \bibinfo{year}{2002}\natexlab{}.
\newblock \showarticletitle{Mimicry attacks on host-based intrusion detection
  systems}. In \bibinfo{booktitle}{\emph{Proceedings of the 9th ACM Conference
  on Computer and Communications Security}}. ACM, \bibinfo{pages}{255--264}.
\newblock


\bibitem[\protect\citeauthoryear{Wang, Chen, Abdou, and Jana}{Wang
  et~al\mbox{.}}{2018a}]%
        {wang2018mixtrain}
\bibfield{author}{\bibinfo{person}{Shiqi Wang}, \bibinfo{person}{Yizheng Chen},
  \bibinfo{person}{Ahmed Abdou}, {and} \bibinfo{person}{Suman Jana}.}
  \bibinfo{year}{2018}\natexlab{a}.
\newblock \showarticletitle{MixTrain: Scalable Training of Formally Robust
  Neural Networks}.
\newblock \bibinfo{journal}{\emph{arXiv preprint arXiv:1811.02625}}
  (\bibinfo{year}{2018}).
\newblock


\bibitem[\protect\citeauthoryear{Wang, Pei, Justin, Yang, and Jana}{Wang
  et~al\mbox{.}}{2018b}]%
        {shiqi2018efficient}
\bibfield{author}{\bibinfo{person}{Shiqi Wang}, \bibinfo{person}{Kexin Pei},
  \bibinfo{person}{Whitehouse Justin}, \bibinfo{person}{Junfeng Yang}, {and}
  \bibinfo{person}{Suman Jana}.} \bibinfo{year}{2018}\natexlab{b}.
\newblock \showarticletitle{Efficient Formal Safety Analysis of Neural
  Networks}.
\newblock \bibinfo{journal}{\emph{Advances in Neural Information Processing
  Systems (NIPS)}} (\bibinfo{year}{2018}).
\newblock


\bibitem[\protect\citeauthoryear{Wang, Pei, Justin, Yang, and Jana}{Wang
  et~al\mbox{.}}{2018c}]%
        {reluval2018}
\bibfield{author}{\bibinfo{person}{Shiqi Wang}, \bibinfo{person}{Kexin Pei},
  \bibinfo{person}{Whitehouse Justin}, \bibinfo{person}{Junfeng Yang}, {and}
  \bibinfo{person}{Suman Jana}.} \bibinfo{year}{2018}\natexlab{c}.
\newblock \showarticletitle{Formal Security Analysis of Neural Networks using
  Symbolic Intervals}.
\newblock \bibinfo{journal}{\emph{27th USENIX Security Symposium}}
  (\bibinfo{year}{2018}).
\newblock


\bibitem[\protect\citeauthoryear{Wang, Zhang, Xu, Lin, Jana, Hsieh, and
  Kolter}{Wang et~al\mbox{.}}{2021}]%
        {wang2021beta}
\bibfield{author}{\bibinfo{person}{Shiqi Wang}, \bibinfo{person}{Huan Zhang},
  \bibinfo{person}{Kaidi Xu}, \bibinfo{person}{Xue Lin}, \bibinfo{person}{Suman
  Jana}, \bibinfo{person}{Cho-Jui Hsieh}, {and} \bibinfo{person}{J~Zico
  Kolter}.} \bibinfo{year}{2021}\natexlab{}.
\newblock \showarticletitle{Beta-CROWN: Efficient Bound Propagation with
  Per-neuron Split Constraints for Complete and Incomplete Neural Network
  Verification}.
\newblock \bibinfo{journal}{\emph{arXiv preprint arXiv:2103.06624}}
  (\bibinfo{year}{2021}).
\newblock


\bibitem[\protect\citeauthoryear{Wehenkel and Louppe}{Wehenkel and
  Louppe}{2019}]%
        {wehenkel2019unconstrained}
\bibfield{author}{\bibinfo{person}{Antoine Wehenkel} {and}
  \bibinfo{person}{Gilles Louppe}.} \bibinfo{year}{2019}\natexlab{}.
\newblock \showarticletitle{Unconstrained monotonic neural networks}. In
  \bibinfo{booktitle}{\emph{Advances in Neural Information Processing
  Systems}}.
\newblock


\bibitem[\protect\citeauthoryear{Weng, Zhang, Chen, Song, Hsieh, Daniel,
  Boning, and Dhillon}{Weng et~al\mbox{.}}{2018a}]%
        {weng2018towards}
\bibfield{author}{\bibinfo{person}{Lily Weng}, \bibinfo{person}{Huan Zhang},
  \bibinfo{person}{Hongge Chen}, \bibinfo{person}{Zhao Song},
  \bibinfo{person}{Cho-Jui Hsieh}, \bibinfo{person}{Luca Daniel},
  \bibinfo{person}{Duane Boning}, {and} \bibinfo{person}{Inderjit Dhillon}.}
  \bibinfo{year}{2018}\natexlab{a}.
\newblock \showarticletitle{Towards fast computation of certified robustness
  for relu networks}. In \bibinfo{booktitle}{\emph{International Conference on
  Machine Learning (ICML)}}.
\newblock


\bibitem[\protect\citeauthoryear{Weng, Zhang, Chen, Yi, Su, Gao, Hsieh, and
  Daniel}{Weng et~al\mbox{.}}{2018b}]%
        {weng2018evaluating}
\bibfield{author}{\bibinfo{person}{Tsui-Wei Weng}, \bibinfo{person}{Huan
  Zhang}, \bibinfo{person}{Pin-Yu Chen}, \bibinfo{person}{Jinfeng Yi},
  \bibinfo{person}{Dong Su}, \bibinfo{person}{Yupeng Gao},
  \bibinfo{person}{Cho-Jui Hsieh}, {and} \bibinfo{person}{Luca Daniel}.}
  \bibinfo{year}{2018}\natexlab{b}.
\newblock \showarticletitle{Evaluating the robustness of neural networks: An
  extreme value theory approach}. In \bibinfo{booktitle}{\emph{International
  Conference on Learning Representations (ICLR)}}.
\newblock


\bibitem[\protect\citeauthoryear{Wong and Kolter}{Wong and Kolter}{2018}]%
        {wong2018provable}
\bibfield{author}{\bibinfo{person}{Eric Wong} {and} \bibinfo{person}{Zico
  Kolter}.} \bibinfo{year}{2018}\natexlab{}.
\newblock \showarticletitle{Provable defenses against adversarial examples via
  the convex outer adversarial polytope}. In
  \bibinfo{booktitle}{\emph{International Conference on Machine Learning}}.
  \bibinfo{pages}{5283--5292}.
\newblock


\bibitem[\protect\citeauthoryear{Wong, Schmidt, Metzen, and Kolter}{Wong
  et~al\mbox{.}}{2018}]%
        {wong2018scaling}
\bibfield{author}{\bibinfo{person}{Eric Wong}, \bibinfo{person}{Frank Schmidt},
  \bibinfo{person}{Jan~Hendrik Metzen}, {and} \bibinfo{person}{J~Zico Kolter}.}
  \bibinfo{year}{2018}\natexlab{}.
\newblock \showarticletitle{Scaling provable adversarial defenses}.
\newblock \bibinfo{journal}{\emph{Advances in Neural Information Processing
  Systems (NIPS)}} (\bibinfo{year}{2018}).
\newblock


\bibitem[\protect\citeauthoryear{Xu, Zhang, Wang, Wang, Jana, Lin, and
  Hsieh}{Xu et~al\mbox{.}}{2021}]%
        {xu2020fast}
\bibfield{author}{\bibinfo{person}{Kaidi Xu}, \bibinfo{person}{Huan Zhang},
  \bibinfo{person}{Shiqi Wang}, \bibinfo{person}{Yihan Wang},
  \bibinfo{person}{Suman Jana}, \bibinfo{person}{Xue Lin}, {and}
  \bibinfo{person}{Cho-Jui Hsieh}.} \bibinfo{year}{2021}\natexlab{}.
\newblock \showarticletitle{Fast and complete: Enabling complete neural network
  verification with rapid and massively parallel incomplete verifiers}.
\newblock \bibinfo{journal}{\emph{International Conference on Learning
  Representations (ICLR)}} (\bibinfo{year}{2021}).
\newblock


\bibitem[\protect\citeauthoryear{Yang, Duan, Hu, Salman, Razenshteyn, and
  Li}{Yang et~al\mbox{.}}{2020}]%
        {yang2020randomized}
\bibfield{author}{\bibinfo{person}{Greg Yang}, \bibinfo{person}{Tony Duan},
  \bibinfo{person}{J~Edward Hu}, \bibinfo{person}{Hadi Salman},
  \bibinfo{person}{Ilya Razenshteyn}, {and} \bibinfo{person}{Jerry Li}.}
  \bibinfo{year}{2020}\natexlab{}.
\newblock \showarticletitle{Randomized smoothing of all shapes and sizes}. In
  \bibinfo{booktitle}{\emph{International Conference on Machine Learning
  (ICML)}}. PMLR.
\newblock


\bibitem[\protect\citeauthoryear{Yao, Ryan, Wong, Jana, and Gu}{Yao
  et~al\mbox{.}}{2020}]%
        {yao2020learning}
\bibfield{author}{\bibinfo{person}{Jianan Yao}, \bibinfo{person}{Gabriel Ryan},
  \bibinfo{person}{Justin Wong}, \bibinfo{person}{Suman Jana}, {and}
  \bibinfo{person}{Ronghui Gu}.} \bibinfo{year}{2020}\natexlab{}.
\newblock \showarticletitle{Learning Nonlinear Loop Invariants with Gated
  Continuous Logic Networks}. In \bibinfo{booktitle}{\emph{Proceedings of the
  41st ACM SIGPLAN Conference on Programming Language Design and
  Implementation}}.
\newblock


\bibitem[\protect\citeauthoryear{Zhang, Chen, Xiao, Gowal, Stanforth, Li,
  Boning, and Hsieh}{Zhang et~al\mbox{.}}{2020}]%
        {zhang2020towards}
\bibfield{author}{\bibinfo{person}{Huan Zhang}, \bibinfo{person}{Hongge Chen},
  \bibinfo{person}{Chaowei Xiao}, \bibinfo{person}{Sven Gowal},
  \bibinfo{person}{Robert Stanforth}, \bibinfo{person}{Bo Li},
  \bibinfo{person}{Duane Boning}, {and} \bibinfo{person}{Cho-Jui Hsieh}.}
  \bibinfo{year}{2020}\natexlab{}.
\newblock \showarticletitle{Towards stable and efficient training of verifiably
  robust neural networks}.
\newblock \bibinfo{journal}{\emph{International Conference on Learning
  Representations (ICLR)}} (\bibinfo{year}{2020}).
\newblock


\bibitem[\protect\citeauthoryear{Zhang, Weng, Chen, Hsieh, and Daniel}{Zhang
  et~al\mbox{.}}{2018}]%
        {zhang2018efficient}
\bibfield{author}{\bibinfo{person}{Huan Zhang}, \bibinfo{person}{Tsui-Wei
  Weng}, \bibinfo{person}{Pin-Yu Chen}, \bibinfo{person}{Cho-Jui Hsieh}, {and}
  \bibinfo{person}{Luca Daniel}.} \bibinfo{year}{2018}\natexlab{}.
\newblock \showarticletitle{Efficient neural network robustness certification
  with general activation functions}.
\newblock \bibinfo{journal}{\emph{arXiv preprint arXiv:1811.00866}}
  (\bibinfo{year}{2018}).
\newblock


\bibitem[\protect\citeauthoryear{Zhang and Evans}{Zhang and Evans}{2019}]%
        {zhang2018cost}
\bibfield{author}{\bibinfo{person}{Xiao Zhang} {and} \bibinfo{person}{David
  Evans}.} \bibinfo{year}{2019}\natexlab{}.
\newblock \showarticletitle{{Cost-Sensitive Robustness against Adversarial
  Examples}}.
\newblock \bibinfo{journal}{\emph{International Conference on Learning
  Representations (ICLR)}} (\bibinfo{year}{2019}).
\newblock


\end{thebibliography}

\appendix
\begin{figure*}[ht!]
    \centering
    \includegraphics[width=\textwidth]{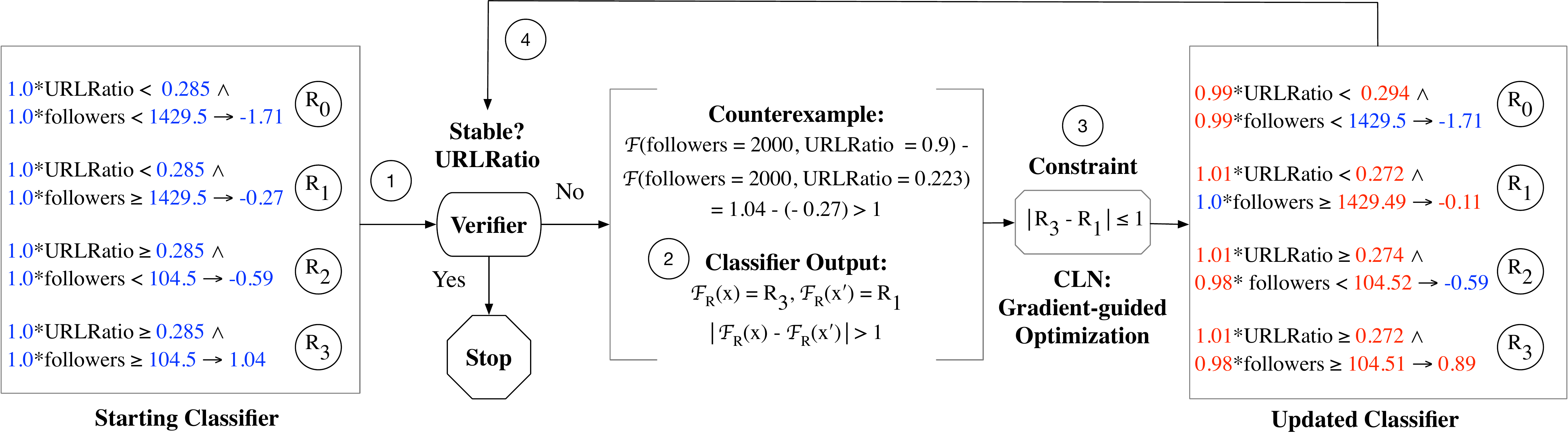}
    \caption{One CEGIS iteration to train stability property for the Twitter account classifier.
    We specify the classifier's output score to change at most by one when the URLRatio feature
    is arbitrarily perturbed. Multiple weights of the classifiers are updated by gradient-guided optimization,
    and the classifier after training no longer forms a tree structure.}
    \label{fig:stability_demo}
\end{figure*}

\section{Stability for Twitter Account Classifier}
\label{appendix:Stability for Twitter Account Classifier}

To classify Twitter accounts that broadcast spam URLs,
we can use the number of followers
and the ratio of posted URLs over total number of tweets as features~\cite{lee2011seven}.
It is hard for spammers to obtain large amount of followers, and
they are likely to post more URLs than benign users.
We specify the URLRatio feature to be stable, such that arbitrarily changing
the feature will not change the classifier's output by more than 1.

Figure~\ref{fig:stability_demo} shows one CEGIS iteration to train
the stability property. The starting classifier is a decision tree.
For example, ``$1.0*\texttt{URLRatio} < 0.285 \land 1.0*\texttt{followers} < 1429.5 \rightarrow -1.71$''
means that if the URLRatio and the number of followers both satisfy these inequalities,
the clause is true and returns $-1.71$, value of the variable $R_0$.
Otherwise, the clause returns $0$.
We take the sum of return values from all clauses to be the classification score.
One CEGIS iteration goes through the following four steps.

Step \textcircled{1}: We ask the verifier whether the URLRatio feature is stable.
If the verifier can verify the stability property,
we stop here. If not, the verifier generates a counterexample that violates the property.
Here, the counterexample shows that if the number of followers is 2000, and if the
URLRatio feature changes from $0.9$ to $0.223$, the classifier's output changes by 1.31,
which violates the stability property. 

Step \textcircled{2}:
Using the sum of true clauses for each input, we represent $x, x^\prime$
as $\mathcal{F}_R(x) = R_3$, and $\mathcal{F}_R(x^\prime) = R_1$.


Step \textcircled{3}: We construct the constraint to eliminate the counterexample. In this case,
we want the difference between the output for $x$
and the output for $x^\prime$ to be bounded by 1, i.e., $\abs{R_3 - R_1} \leq 1$.
Then, we smooth the classifier using CLN~\cite{ryan2019cln2inv, yao2020learning},
train the weights using projected gradient descent with the constraint.
After one epoch, we have updated the classifier in the
rightmost box of Figure~\ref{fig:stability_demo}.
The red weights of the model are updated by gradient descent.
Note that the classifier no longer follows a tree structure.

Lastly, we repeat this process until the classifier is verified to satisfy the property (Step \textcircled{4}).
In this example, the updated classifier still does not satisfy stability, and we will
go through more CEGIS iterations to update it.

\section{Proof}
\label{appendix:proof}

\begin{lemma}
If a classifier satisfies Property 3a, then it also satisfies Property 3.
\label{lemma1}
\end{lemma}

\begin{proof}
$\forall x,x' \in \mathbb{R}^n . [\forall i \notin J . x_i = x'_i] \land g(\mathcal{F}(x)) \ge \delta$,
we have $\mathcal{F}(x) \ge g^{-1}(\delta)$.
Since $\mathcal{F}$ satisfies Property 3a, then
we also have
$\mathcal{F}(x) - \mathcal{F}(x') \le g^{-1}(\delta)$.
Therefore, $\mathcal{F}(x') \ge \mathcal{F}(x) - g^{-1}(\delta) \ge 0$.

\end{proof}

\section{Measurement Results}
\label{appendix:Measurement Results}

\subsection{LenScreenName Feature}
\label{appendix:LenScreenName Feature for Twitter Spam Account Dataset}
\begin{figure}[h!]
    \centering
    \includegraphics[width=0.6\columnwidth]{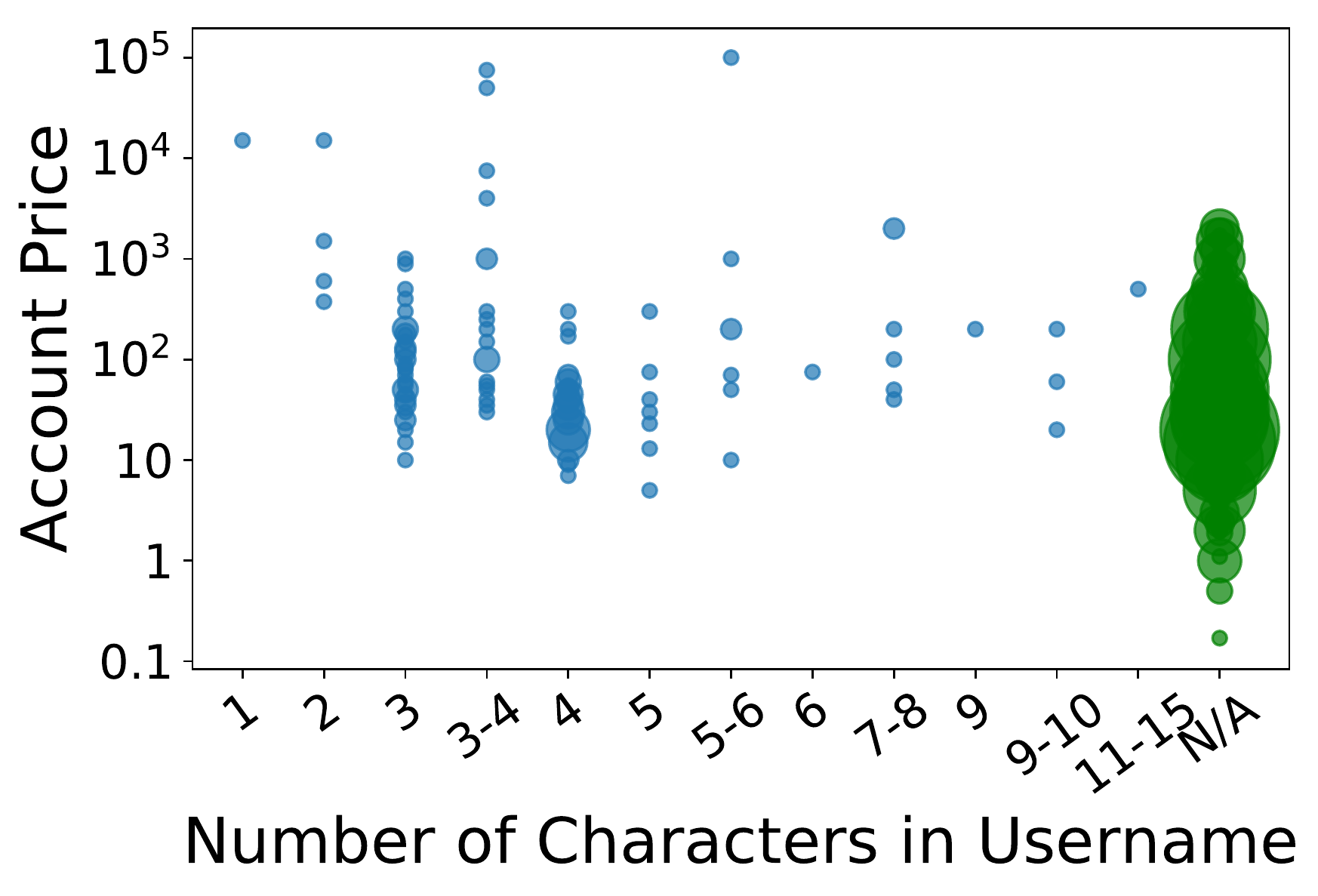}
    \caption{Price (\$, log scale) of Twitter accounts with different number of characters.}
    \label{fig:charprice_plot}
\end{figure}

We measure the economic cost for attackers to perturb
the LenScreenName feature from the Twitter spam account dataset.
We extract the number of characters information from 6,125 for-sale Twitter account posts, and measure the price for accounts with different username length. 
or unspecificed characters.
In Figure~\ref{fig:charprice_plot},
we plot the price of accounts according
to the username length. If the post says ``3 or 4 characters",
we plot the price under ``3-4" category. The majority of
accounts are under ``N/A'' category, where the sellers do not
mention the length of username, but emphasize other attributes
such as number of followers. Overall, if the username
length has at most 4 characters, it affects the account
price more than longer username.

\subsection{CCSize Feature}
\label{appendix:CCSize Feature for Twitter Spam URL Dataset}

\begin{figure}[h!]
    \centering
    \includegraphics[width=0.6\columnwidth]{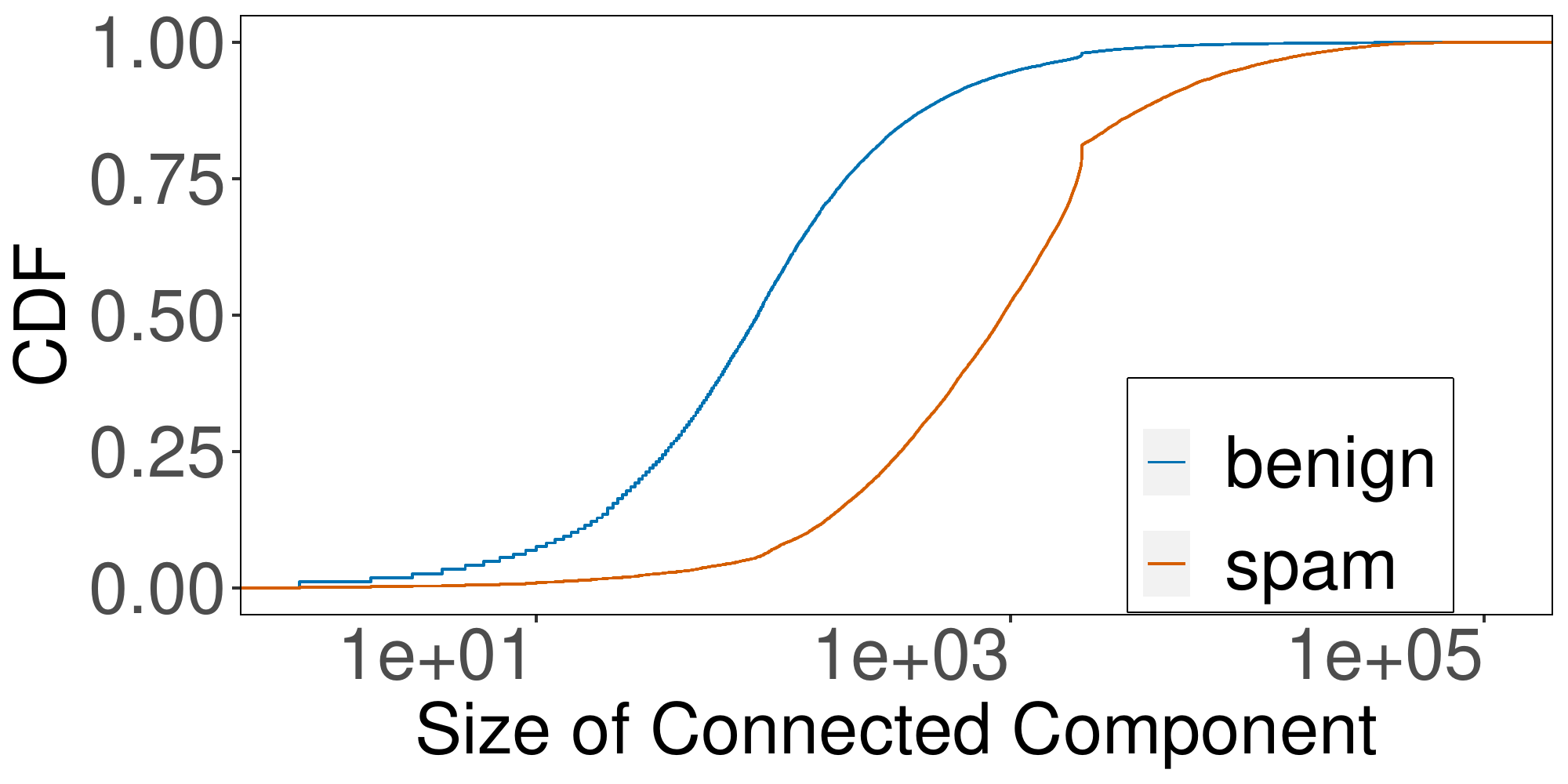}
    \caption{CDF of
    \# of IPs in a connected component containing
    a given URL. Spam URLs tend to be in larger
    connected components.}
    \label{fig:cdf_ccsize}
\end{figure}

We measure the distribution of CCSize feature from
the Twitter spam URL dataset.
The CCSize feature counts
the number of IP nodes in the connected component of the posted URL.
Since spammers reuse redirectors, their URLs often belong to the same large
connected component and result in a larger CCSize feature value,
compared to benign URLs,
as shown in Figure~\ref{fig:cdf_ccsize}. A larger CCSize value
indicates that more resources are being reused, and the initial URL is more suspicious. Therefore, we specify CCSize feature
to be monotonically increasing.

\begin{table}[t!]
  \centering
  \begin{tabular}{l | c | c}
    \hline
    \textbf{Property} & \textbf{Training Constraints} & \textbf{Robustness} \\
    \hline \hline
    Stability & Smaller $c_\text{stability}$ & Stronger \\ \hline
    High Confidence & Smaller $\delta$ & Stronger \\ \hline
    Redundancy & Smaller $\delta$ & Stronger \\ \hline
    Small Neighborhood & Smaller $c$ given fixed $\epsilon$ & Stronger \\ \hline
  \end{tabular}
  \caption{Robustness controlled by hyperparameters.}
  \label{tab:params}
\end{table}

\begin{table}[t!]
  \centering
  \begin{tabular}{c | c | c}
    \hline
    \textbf{Stable Constant $c_\text{stability}$} & \textbf{TPR} & \textbf{FPR} \\
    \hline \hline
    8 & 0.86 & 0.021 \\ \hline
    4 & 0.835 & 0.031 \\ \hline
    2 & 0.833 & 0.029 \\ \hline
  \end{tabular}
  \caption{TPR and FPR of Twitter spam account classifiers trained with the stability property.}
  \label{tab:twitter_stable}
\end{table}

\section{Classification Features}
\label{appendix:Classification Features}

Table~\ref{tab:features} lists all the features
for detecting cryptojacking, Twitter spam accounts,
and Twitter spam URLs.

\begin{table*}[!bt]
	\centering
	\small
	\begin{tabular}{c | c | l | c | c}
		\hline
		\textbf{Dataset} & \textbf{Feature Name} & \textbf{Description} & \textbf{Monotonic} & \textbf{Low-cost} \\ \hline\hline
		\multirow{7}{*}{Cryptojacking} & websocket & Use WebSocket APIs for network communication &  Increasing &   \\ \cline{2-5}
		& wasm & Uses WebAssembly to execute code in browsers are near native speed & Increasing &  \\ \cline{2-5}
		& hash function & Use one of the hash functions on a curated list & Increasing & yes \\ \cline{2-5}
		& webworkers & The number of web workers threads for running concurrent tasks & Increasing &  \\ \cline{2-5}
		& messageloop load & The number of MessageLoop events for thread management & Increasing &  \\ \cline{2-5}
		& postmessage load & The number of PostMessage events for thread job reporting & Increasing &  \\ \cline{2-5}
		& parallel functions & Run the same tasks in multiple threads  & Increasing &  \\ \hline \hline
		
		\multirow{15}{*}{\begin{tabular}[c]{@{}c@{}}Twitter Spam \\ Accounts \end{tabular}}
		& LenScreenName & The number of characters in the account user name & Increasing & yes ($\geq 5$ char)  \\ \cline{2-5}
		& LenProfileDescription & The number of characters in the profile description & & yes  \\ \cline{2-5}
		& AgeDays & The age of the account in days & Decreasing &  \\ \cline{2-5}
		& NumFollowings & The number of other users an account follows & Increasing &  \\ \cline{2-5}
		& NumFollowers & The number of followers for an account & Decreasing &  \\ \cline{2-5}
		& Ratio\_Following\_Followers & The ratio of NumFollowings divided by NumFollowers & &  \\ \cline{2-5}
		& StdFollowing & Standard deviation of NumFollowings over different days & & \\ \cline{2-5}
		& ChangeRateFollowing & The averaged difference for NumFollowings between consecutive days & & \\ \cline{2-5}
		& NumTweets & Total number of tweets over seven months &  & yes \\ \cline{2-5}
		& NumDailyTweets & Average number of daily tweets &  & yes \\ \cline{2-5}
		& TweetLinkRatio & Ratio of tweets containing links over total number of tweets & Increasing & yes \\ \cline{2-5}
		& TweetUniqLinkRatio & Ratio of tweets containing unique links over total number of tweets & Increasing & yes  \\ \cline{2-5}
		& TweetAtRatio & Ratio of tweets containing `@' over total number of tweets & & yes \\ \cline{2-5}
		& TweetUniqAtRatio & Ratio of tweets with unique `@' username over total number of tweets & & yes  \\ \cline{2-5}
		& PairwiseTweetSimilarity & Normalized avg num of common chars in pairwise tweets for a user & &  \\ \hline \hline
		 
		& \multicolumn{4}{l}{Shared Resources-driven} \\ \cline{2-5}
		\multirow{29}{*}{\begin{tabular}[c]{@{}c@{}}Twitter Spam \\ URLs \end{tabular}}
		& EntryURLid & In degree of the largest redirector in the connected component & Increasing & \\ \cline{2-5}
		& AvgURLid & Average in degree of URL nodes in the redirection chain & Increasing & \\ \cline{2-5}
		& ChainWeight & Total frequency of edges in the redirection chain & Increasing & \\ \cline{2-5}
		& CCsize & Number of nodes in the connected component & Increasing & \\ \cline{2-5}
		& CCdensity & Edge density of the connected component & & \\ \cline{2-5}
		& MinRCLen & Minimum length of the redirection chains in the connected component & Increasing & \\ \cline{2-5}
		& AvgLdURLDom & Avg \# of domains for landing URL IPs in the connected component & Increasing & \\ \cline{2-5}
		& AvgURLDom & Average \# of domains for the IPs in the redirection chain & Increasing & \\ \cline{2-5}
		& \multicolumn{4}{l}{Heterogeneity-driven} \\ \cline{2-5}
		& GeoDist & Total geographical distance (km) traversed by the redirection chain & & \\ \cline{2-5}
		& CntContinent & Number of unique continents in the redirection chain & & \\ \cline{2-5}
		& CntCountry & Number of unique countries in the redirection chain  & & \\ \cline{2-5}
		& CntIP & Number of unique IPs in the redirection chain & & \\ \cline{2-5}
		& CntDomain & Number of unique domains in the redirection chain & & \\ \cline{2-5}
		& CntTLD & Number of unique top-level domains in the redirection chain & & \\ \cline{2-5}
		& \multicolumn{4}{l}{Flexibility-driven} \\ \cline{2-5}
		& ChainLen & Length of the redirection chain & & \\ \cline{2-5}
		& EntryURLDist & Distance from the initial URL to the largest redirector & & \\ \cline{2-5}
		& CntInitURL & Number of initial URLs in the connected component & & \\ \cline{2-5}
		& CntInitURLDom & Total domain name number in the initial URLs & & \\ \cline{2-5}
		& CntLdURL & Number of final landing URLs in the redirection chain & & \\ \cline{2-5}
		& AvgIPperURL & Average IP number per URL in the connected component & & \\ \cline{2-5}
		& AvgIPperLdURL & Average IP number per landing URL in the connected component & & \\ \cline{2-5}
		& \multicolumn{4}{l}{Tweet Content} \\ \cline{2-5}
		& Mention Count & Number of ‘@’  that mentions other users & & yes \\ \cline{2-5}
		& Hashtag Count & Number of hashtags & & yes \\ \cline{2-5}
		& Tweet Count & Number of tweets made by the user account for this tweet & & yes \\ \cline{2-5}
		& URL Percent & Percentage of posts from the same user that contain a URL & & yes \\ \hline
	\end{tabular} 
	\caption{Classification features for three datasets. For each feature, we also mark the monotonic direction if we specify the monotonicity property, and whether we specify the feature to be low cost.}
	\label{tab:features}
\end{table*}

\section{Heuristics}
\label{appendix:Heuristics}

The time to solve for counterexamples is the bottleneck in training.
Therefore, we implement the following heuristics to improve
the training efficiency:
\begin{itemize}[leftmargin=*]
\item We exponentially increase the time out for the solver, starting
from 30s. If we find at least one counterexample
within each CEGIS iteration, we add the constraint(s) to eliminate
the counterexample(s) and proceed with CLN training with the constraints.
We increase the timeout if the solver could not find any counterexample
fast enough in an iteration, and if
it also could not verify that the classifier satisfies all the specified properties.
\item We implement property boosting as an option to train monotonicity and stability. Property boosting means that we only train the property
for the newly added sub-classifier, and keep the previous sub-classifiers
fixed. This works for properties that can be satisfied if every sub-classifiers
also satisfy the sub-properties, since our ensemble is a sum ensemble.
If every sub-classifier is monotonic for a given
feature, the ensemble classifier is also monotonic. Similarly, if every sub-classifier is stable for a given feature by a stable constant $\frac{c}{B}$,
the ensemble classifier is stable under stable constant $c$.
\item We use feature scheduling to train
the high confidence property. Specifically, to run 10 rounds of boosting
for either Twitter spam account or Twitter spam URL detection classifiers,
we first boost 6 decision trees as the base model without any low-cost features.
This makes sure that the base model naturally satisfies the high confidence
property. Then, for the remaining 4 rounds, we use all features to boost
new trees and fix the properties for the entire classifier. 
\item When training all the five properties (monotonicity, stability, high confidence, redundancy and small neighborhood) for the Twitter spam account detection,
we use the following property scheduling to boost 6 rounds.
For the first round, we use features that don't involve any property
to construct a base classifier, so it naturally satisfies all properties.
In the 2nd and 3rd round, we use all features excluding low-cost ones,
so we get high confidence and redundancy for free for these rounds.
In the next two rounds, we use all features excluding monotonic ones,
so we get monotonicity for free for these rounds. In the last round,
we use all features and fix all five properties.
\end{itemize}

Property boosting, feature scheduling, and property scheduling
reduce the size of the integer linear program, which makes it easier to be solved.

%

\section{Hyperparameters}
\label{appendix:Hyperparameters}

Enforcing stronger robustness decreases true positive rate. The hyperparameters control this tradeoff.
In particular, Table~\ref{tab:params} shows how the strength of robustness changes
as different hyperparameters change for all proposed properties except monotonicity (we don't have such
a hyperparameter for monotonicity). For example, to demonstrate the tradeoff,
we trained three Twitter spam account classifiers with the stability property,
where each one has a different stable constant. Table~\ref{tab:twitter_stable}
shows that training with a smaller stable constant $c_\text{stability}$ gives
us a verifiably robust model with stronger robustness but lower true positive rate.

\section{Obtaining More Properties}
\label{appendix:Obtaining More Properties}

Table~\ref{tab:twitter_results} shows that training a classifier with one property sometimes obtains another property.
For the Twitter spam account detection classifiers,
enforcing one of the high confidence, redundancy,
and small neighborhood properties can obtain at least
a second property. For example, the Logic Ensemble Redundancy model
has obtained stability and high confidence properties.
Since we use the same set of low-cost features
to define the high confidence and redundancy properties,
the redundancy property is strictly stronger than
the high confidence property. In other words,
if the attacker have to perturb one low-cost feature
from at least two different groups to evade the classifier (redundancy),
they cannot evade the classifier by perturbing only one low-cost feature (high confidence).
For the largest Twitter spam URL detection dataset,
the Logic Ensemble Stability model also satisfies the small neighborhood
property, and the Logic Ensemble High confidence model also satisfies
the stability property.

\section{Logic Ensemble Combined Model}
\label{appendix:Logic Ensemble Combined Model}

Table~\ref{tab:twitter_results} shows that,
for Twitter spam account detection,
the Logic Ensemble Combined with all five properties
has higher AUC than the
Logic Ensemble Monotonicity model trained with only one property.
This is because we use property scheduling for Logic Ensemble Combined (Appendix~\ref{appendix:Heuristics}),
such that for each round before the last round, our classifier satisfies some properties for free.
We could improve the performance of the Logic Ensemble Monotonicity
by similar feature scheduling technique, such as boosting first four
rounds of model with non-monotonic features first, and then train
with all features for later rounds.

\end{document}